\documentclass[12pt, a4paper]{article}
\usepackage[a4paper, margin=2.5cm]{geometry}
\usepackage{amsmath}
\usepackage{hyperref}
\usepackage{amsthm}
\usepackage{nccmath}
\usepackage{proba}
\usepackage{clrscode}
\usepackage{xspace}
\usepackage{caption}
\usepackage{amsfonts}
\usepackage{graphicx}
\usepackage{subfig}
\usepackage{authblk}
\usepackage{algorithm}
\usepackage{algpseudocode}
\usepackage{mathtools}

\algnewcommand\algorithmicswitch{\textbf{switch}}
\algnewcommand\algorithmiccase{\textbf{case}}
\algnewcommand\algorithmicassert{\texttt{assert}}
\algnewcommand\Assert[1]{\State \algorithmicassert(#1)}%
\algdef{SE}[SWITCH]{Switch}{EndSwitch}[1]{\algorithmicswitch\ #1\ \algorithmicdo}{\algorithmicend\ \algorithmicswitch}%
\algdef{SE}[CASE]{Case}{EndCase}[1]{\algorithmiccase\ #1}{\algorithmicend\ \algorithmiccase}%
\algtext*{EndSwitch}%
\algtext*{EndCase}%

\newenvironment{tight_enumerate}{
\begin{enumerate}
  \setlength{\itemsep}{0pt}
  \setlength{\parskip}{0pt}
}{\end{enumerate}}

\graphicspath{{figures/}}

\title{Distance Measures for Embedded Graphs} 
\author[1]{Hugo A. Akitaya\footnote{Supported by National Science Foundation grants CCF-1422311 and CCF-1423615, and the Science Without Borders scholarship program.}}
\author[2]{Maike Buchin}
\author[2]{Bernhard Kilgus\footnote{Supported by the Deutsche Forschungsgemeinschaft (DFG), project BU 2419/3-1}}
\author[2]{Stef Sijben}
\author[3]{Carola Wenk\footnote{Supported by National Science Foundation grant CCF-1618469}}
\affil[1]{Department of Computer Science, Tufts University, Medford, MA, USA\\
	\texttt{hugo.alves\_akitaya@tufts.edu}}
\affil[2]{Department of Mathematics, Ruhr University Bochum, Bochum, Germany\\
	\texttt{Maike.Buchin|Bernhard.Kilgus|Stef.Sijben@rub.de}}
\affil[3]{Department of Computer Science, Tulane University, New Orleans\\
	USA, \texttt{cwenk@tulane.edu}}

\newcommand{\fr}{Fr\'echet distance\xspace}

\newcommand{\wfr}{weak Fr\'echet distance\xspace}

\newcommand{\eps}{\varepsilon}

\newcommand{\dist}{\delta_G}
\newcommand{\distWeak}{\delta_{wG}}
\newcommand{\distDir}{\vec{\delta}_G}
\newcommand{\distDirWeak}{\vec{\delta}_{wG}}

\newcommand{\algoref}[1]{Algorithm~\ref{#1}}
\newcommand{\figref}[1]{Figure~\ref{#1}}
\newcommand{\defref}[1]{Definition~\ref{#1}}
\newcommand{\lemref}[1]{Lemma~\ref{#1}}

\newcommand{\thmref}[1]{Theorem~\ref{#1}}
\newcommand{\secref}[1]{Section~\ref{#1}}

\newtheorem{theorem}{Theorem}
\newtheorem{lemma}{Lemma}

\newtheorem{definition}{Definition}
\newtheorem{problem}{Problem}
\newtheorem{observation}{Observation}

\begin{document}

\maketitle

\begin{abstract}
  We introduce new distance measures for comparing straight-line embedded 
graphs based on the \fr and the weak Fr\'echet distance.
These graph distances are defined using continuous mappings
and thus take the
combinatorial structure as well as the geometric embeddings of the
graphs into account.  We present a general algorithmic approach for
computing these graph distances. Although we show that deciding the
distances is NP-hard for general embedded graphs, we prove that our
approach yields polynomial time algorithms if the graphs are trees,
and for the distance based on the weak \fr if the graphs are planar
embedded.  Moreover, we prove that deciding the distances based on
the \fr remains NP-hard for planar embedded graphs and show how our
general algorithmic approach yields an exponential time algorithm and
a polynomial time approximation algorithm for this case. Our work
combines and extends the work of Buchin et al.~\cite{bsw-dmeg-17} and
Akitaya et al.~\cite{bsw-dmeg-19} presented at EuroCG.

\end{abstract}

\section{Introduction}
There are many applications that work with graphs that are embedded in
Euclidean space.  One task that arises in such applications is
comparing two embedded graphs.  For instance, the two graphs to be
compared could be two different representations of a geographic
network (e.g., roads or rivers).  Oftentimes these networks are not
isomorphic, nor is one interested in subgraph isomorphism, but one
would like to have a mapping of one graph to the other, and ideally such a mapping would be continuous. For instance,
this occurs when we have a ground truth of a road network and a
simplification or reconstruction of the same network and we would like
to measure the error of the latter.  In this case, a mapping would
identify the parts of the ground truth that are
reconstructed/simplified and would allow to study the local
error.

We present new graph distance measures that are well-suited for comparing such graphs.
Our distance measures are natural generalizations of
the \fr~\cite{alt1995computing} to graphs and require a continuous mapping, but they don't require graphs to be homeomorphic.
One graph is mapped continuously to a portion of the other, in such a way that edges are mapped to paths in
the other graph.  The graph distance is then defined as the maximum of
the strong (or weak) Fr\'echet distances between the edges and the paths they are
mapped to.
This results in a directed or asymmetric notion of distance, and we define the
corresponding undirected distances as the maximum of both directed distances. The
directed distances naturally arise when seeking to measure subgraph
similarity, which requires mapping one graph to a subgraph of the
other.

For comparing two not necessarily isomorphic graphs only few measures
were known previously. One such measure is the traversal distance
suggested by Alt et al.~\cite{alt2003262} and another is the geometric
edit distance suggested by Cheong et al.~\cite{cgkss-msgg-09}.
The traversal distance converts graphs into curves by traversing the
graphs continuously and comparing the resulting curves using
the \fr. It is also a directed distance that compares the traversal of
one graph with the traversal of a part of the other graph. However, an
explicit mapping between the two graphs is not established, and part of
the connectivity information of the graphs is lost due to the conversion to curves.
%
The geometric edit distance minimizes the cost of edit operations from
one graph to another, where cost is measured by Euclidean lengths and
distances. But again, connectivity is not well maintained.
Figure~\ref{fig:newexamples} shows some examples of graphs
where our graph distances, the traversal distance, and the geometric
edit distance differ. In particular, only our graph distances capture the
difference in connectivity between graphs $G_1$ and $G_2$, as well as between $H_1$
and $H_2$.

\begin{figure}[b]
\centering
 \includegraphics[width=0.8\textwidth]{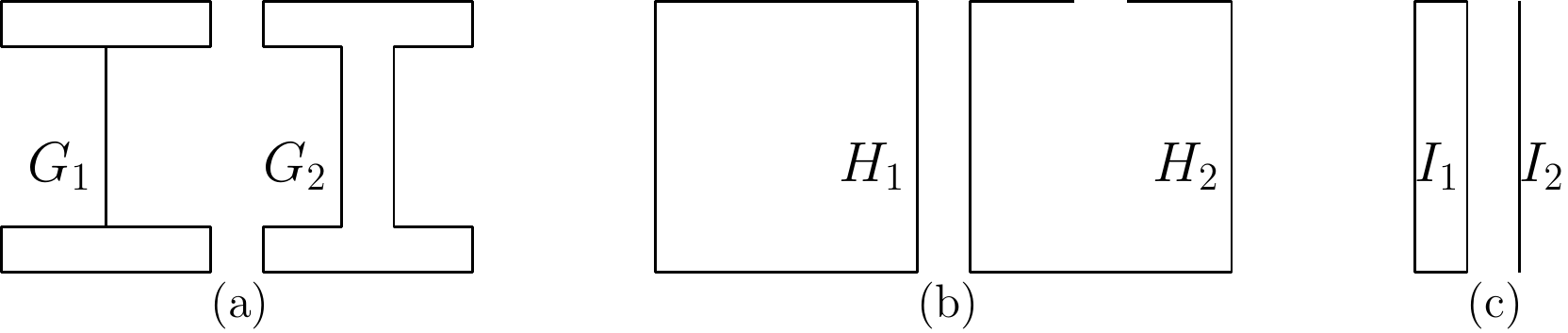}
 \caption{Examples where our graph distances, the traversal
   distance, and the geometric edit distance differ.
For clarity the graphs are shown side-by-side, but in the embedding they lie on top of each other.
   %
   (a): Graphs $G_1$ and $G_2$ have large graph distance (because $G_1$ needs to mapped to one side of $G_2$), large edit distance (because a long edge needs to be added), but small traversal distance.
  (b): Graphs $H_1$ and $H_2$ have large graph distance  (because all of $H_1$ needs
   to mapped to only one side of $H_2$), but small traversal distance and small edit
   distance.
  (c): Graphs $I_1$ and $I_2$ have small graph distance and small traversal distance, but 
  a large edit distance (because a long edge needs to be
   added).}  \label{fig:newexamples}
\end{figure}

Our graph distances map one graph onto a subgraph of the
other and they measure the \fr\ between the mapped parts (see \secref{subsec:definitions} for a
formal definition). Hence connectivity information is preserved and an
explicit mapping between the two (sub-)graphs is established.  One
possible application of these new graph distances is the comparison of
geographic networks, for instance evaluating the quality of map
reconstructions and map simplification. In \secref{sec:experiments}, we show some
experimental results on graphs of map reconstructions that illustrate
that our approach considers both, geometry and connectivity.

\paragraph*{Related work}
A few approaches have been proposed in the literature for comparing
geometric embedded graphs.
Subgraph-isomorphism considers only the combinatorial structure of the
graphs and not its geometric embedding.  It is NP-hard to compute in
general, although it can be computed in linear time if both graphs are
planar and the pattern graph has constant size~\cite{eppstein}.
If we consider the graphs as metric spaces, the Gromov-Hausdorff
distance (GH) between two graphs is the minimum Hausdorff distance
between isometric embeddings of the graphs into a common metric
space. While it is unknown how to compute GH for general graphs,
recently Agarwal et al.~\cite{Agarwal:Gromov} gave a polynomial
time approximation algorithm for the GH between a pair of
metric trees. 
We are however interested in measuring the similarity
between two specific embeddings of the graphs.
Armiti et al.~\cite{geometric:graph:matching} suggest a probabilistic
approach for comparing graphs that are not required to be isomorphic, 
using spatial properties of
the vertices and their neighbors. However, they require vertices to be matched to vertices,
which can result
in a large graph distance when an edge in one graph is subdived in the
other graph. Furthermore, the spatial properties used are invariant to
translation and rotation, whereas we consider a fixed embedding.
%
Cheong et al.~\cite{cgkss-msgg-09} proposed the geometric edit
distance for comparing embedded graphs, however it is NP-hard to compute.
Alt et al.~\cite{alt2003262} defined the traversal distance, which is
most similar to our graph distance measures, but it does not preserve
connectivity. See Section~\ref{subsec:travdist} for a detailed
comparison with the traversal distance.

For assessing the quality of map construction algorithms, several
approaches have been proposed.  One approach is to compare all paths
\cite{aw-pbdsmc-15} or random samples of shortest paths
\cite{Karagiorgou:2012:VTD:2424321.2424334}. However, these measures
ignore the local structure of the graphs. In order to capture more
topological information, Biagioni and Eriksson developed a
sampling-based distance~\cite{be-irmgp-12} and Ahmed et
al.\ introduced the local persistent homology
distance~\cite{afw-lphbdbm-14}.
The latter distance measure focuses on comparing the topology and does
not encode geometric distances between the graphs.
The sampling-based distance is not a formally
defined distance measure, and it crucially depends on parameters (in
particular $matched\_distance$, to decide if points are sufficiently
close to be matched); in practice it is unclear how these parameters
should be chosen. However, it captures the number of matched edges,
which is useful when comparing reconstructed road networks.
In contrast to these measures, our graph distances capture more
topology than the path-based distance \cite{aw-pbdsmc-15}, and capture
differences in geometry better than the local persistent homology
distance \cite{afw-lphbdbm-14}.
Also our graph distances are well-defined distance measures that do not
require specific parameters to be set, unlike
\cite{be-irmgp-12}.

\paragraph*{Contributions}
We present new graph distance measures that compare graphs based on their
geometric embeddings while respecting their combinatorial
structure. To the best of our knowledge, our graph distances are the
first to establish a continuous mapping between the embedded
graphs.
In Section~\ref{sec:pre} we define several variants of our graph
distances (weak, strong, directed, undirected) and study their
properties.  In Section~\ref{sec:alg} we develop an algorithmic
approach for computing the graph distances. On the one hand, we prove
that for general embedded graphs, deciding these distances is
NP-hard. On the other hand, we also show that our algorithmic approach
gives polynomial time algorithms in several cases, e.g., when one
graph is a tree.
The most interesting case is when both graphs are plane. Here, we show
that our algorithmic approach yields a quadratic time algorithm for
the weak \fr.  In Section~\ref{sec:plane_graphs} we focus on plane
graphs and the strong \fr. For this case, we show that the problem is
NP-hard, even though it is polynomial time solvable for the \wfr.
Furthermore, we show how to obtain an approximation, that depends on
the angle between incident edges, in polynomial time and an exact
result in exponential time.

\section{Graph Distance Definition and Properties}\label{sec:pre}
Let $G_1=(V_1,E_1)$ and $G_2=(V_2,E_2)$ be two undirected graphs with
vertices embedded as points in $\mathbb{R}^d$ (typically
$\mathbb{R}^2$) 
that are connected by straight-line
edges. We refer to such graphs as {\em (straight-line) embedded
  graphs}.
Generally, we do not require the graphs to be planar. We denote a
crossing free embedding of a planar graph shortly as a {\em plane
  graph}. Note that for plane graphs $G_1$ and $G_2$, crossings
between edges of $G_1$ and edges of $G_2$ are still allowed.

\subsection{Strong and Weak Graph Distance}\label{subsec:definitions}
We define distance measures between embedded graphs that are based on
mapping one graph to the other. We consider a particular type of graph
mappings, as defined below:

\begin{definition}[Graph Mapping]\label{def:graphMapping}
We call a mapping $s\colon G_1\rightarrow G_2$ a \emph{graph mapping}
if
\begin{tight_enumerate}
  \item it maps each vertex $v\in V_1$ to a point $s(v)$ on an edge of $G_2$, and
  \item it maps each edge $\{u,v\}\in E_1$ to a simple path from $s(u)$ to $s(v)$ in the embedding of $G_2$.
\end{tight_enumerate}
\end{definition}

Note that a graph mapping results in a continuous map if we consider
the graphs as topological spaces.
To measure similarity between edges and mapped paths, our graph distances use the
\fr\ or the weak \fr, which are 
%
popular distance measures
for curves \cite{alt1995computing}.  For two curves $f,g\colon [0,1]
\rightarrow \mathbb{R}^d$ their {\em \fr} is defined as
\begin{ceqn}
\begin{align*}
\delta_F(f,g)=\inf_{\sigma\colon [0,1]\rightarrow[0,1]} \ \max_{t\in [0,1]} \ ||f(t)-g(\sigma(t))||,
\end{align*}
\end{ceqn}
where $\sigma$ ranges over orientation preserving homeomorphisms. 
The {\em weak \fr} is
\begin{ceqn}
\begin{align*}
\delta_{wF}(f,g)=\inf_{\alpha,\beta\colon [0,1]\rightarrow[0,1]} \ \max_{t\in [0,1]} \ ||f(\alpha(t))-g(\beta(t))||\;,
\end{align*}
\end{ceqn}
where $\alpha,\beta$ range over all continuous onto functions that
keep the endpoints fixed.
%

Typically, the \fr is illustrated by a man walking his dog. Here, the
\fr equals the shortest length of a leash that allows the man and the
dog to walk on their curves from beginning to end. For the \wfr man
and dog may walk backwards on their curves, for the \fr they may not.
The \fr and \wfr between two polygonal curves of complexity $n$ can be
computed in $O(n^2 \log n)$ time~\cite{alt1995computing}.
Now, we are ready to define our graph distance measures.

\begin{definition}[Graph Distances]\label{def:distances}
  We define the \emph{directed (strong) graph distance} $\distDir$ as
  \begin{ceqn}
\begin{align*}
\distDir (G_1,G_2) = \inf_{s:G_1\rightarrow G_2} \max_{e\in E_1} \delta_F(e,s(e))
\end{align*}
\end{ceqn}
  and the \emph{directed weak graph distance} $\distDirWeak$ as
    \begin{ceqn}
\begin{align*}
\distDirWeak (G_1,G_2) = \inf_{s:G_1\rightarrow G_2} \max_{e\in E_1} \delta_{wF}(e,s(e))\;,
\end{align*}
\end{ceqn}
where $s$ ranges over graph mappings from $G_1$ to $G_2$, and $e$
and its image $s(e)$ are interpreted as curves in the plane.  The {\em undirected graph distances}
are
\begin{displaymath}
\ \dist(G_1,G_2) = \max (\distDir (G_1,G_2), \distDir(G_2,G_1)) \qquad \text{ and }
\end{displaymath}
\begin{displaymath}
\distWeak(G_1,G_2) = \max (\distDirWeak(G_1,G_2), \distDirWeak(G_2,G_1)).
\end{displaymath}
\end{definition}

According to Definition~\ref{def:graphMapping}, a graph mapping $s$
maps each edge of $G_1$ to a simple path $s(e)$ in $G_2$.
This is justified by the following observation: Mapping $e$ to a
non-simple path $s'(e)$, where $s(e)$ and $s'(e)$
have the same endpoints and $s(e) \subset s'(e)$, does not
decrease the (weak) graph distance because $\delta_{(w)F}(e,
s(e)) \leq \delta_{(w)F}(e, s'(e))$.
From this observation also follows that we cannot decrease
$\distDir(G_1, G_2)$ by adding additional vertices to subdivide an edge
$e$ of $G_1$: While the concatenation of the resulting mapped paths in
$G_2$ may not be simple, it can be replaced by the image of the entire
edge $e$, which by the observation has to be simple.

We state a first important property of the graph distances:
\begin{lemma}
  For embedded graphs, the strong graph distances and the weak graph
  distances fulfill the triangle inequality. The undirected distances
  are pseudo-metrics. For plane graphs they are metrics.
\end{lemma}
\begin{proof} 
Symmetry follows immediately for the undirected distances.  The
directed distances fulfill the triangle inequality because we can
concatenate two maps and use the triangle inequality of $\R^d$: Let
$G_1$, $G_2$ and $G_3$ be three embedded graphs. An edge $e$ of $G_1$
is mapped to a simple path $p$ in $G_2$. The segments of $p$ are again
mapped to a sequence of simple paths in $G_3$. Thus, when
concatenating two maps, one possible mapping maps each edge $e$ of
$G_1$ to a sequence $S$ of simple paths in $G_3$. Note, that $S$ need
not be simple. However, in that case we can instead map $e$ to a
shortest path $\hat{p}$ in $S$ from beginning to end. As
$\delta_{(w)F}(e,\hat{p}) \leq \delta_{(w)F}(e,S)$ for each edge of
$G_1$, we have $\distDir(G_1,G_2)+\distDir(G_2,G_3) \geq
\distDir(G_1,G_3)$ and $\distDirWeak(G_1,G_2)+\distDirWeak(G_2,G_3)
\geq \distDir(G_1,G_3)$ by definition of the directed (weak) graph
distance as the maximum \fr\ of an edge and its mapping.  Analogously,
the undirected distances fulfill the triangle inequality as well.

For plane graphs, their (weak) graph distance is zero iff their
embeddings are the same, hence the distances are metrics.  If the
(weak) graph distance is zero, every edge needs to be mapped to
itself, hence the embeddings are the same.  If on the other hand, the
embeddings are the same, a graph mapping may map every edge to itself
in the embedding.  Since there are no intersections or overlapping
vertices, this mapping is continuous in the target graph, and the
distance is zero.
\end{proof}


Note that for non-plane graphs the (weak) graph distance does not
fulfill the identity of indiscernibles. For example, if $G_1$ consists
of two crossing line segment edges, and $G_2$ has visually the same
embedding but consists of four edges and includes the intersection
point as a vertex, then both,
$\distDir(G_1,G_2)=\distDirWeak(G_1,G_2)=0$ and
$\distDir(G_2,G_1)=\distDirWeak(G_2,G_1)=0$, and therefore
$\dist(G_1,G_2)=\distWeak(G_1,G_2)=0$.
Also note that we do not require graph mappings to be injective or
surjective.  And an optimal graph mapping from $G_1$ to $G_2$ may be
very different from an optimal graph mapping from $G_2$ to $G_1$. See
Figure~\ref{fig:examples} for examples of graphs and their graph
distances.

\begin{figure}[]
\centering
\includegraphics[scale=0.8]{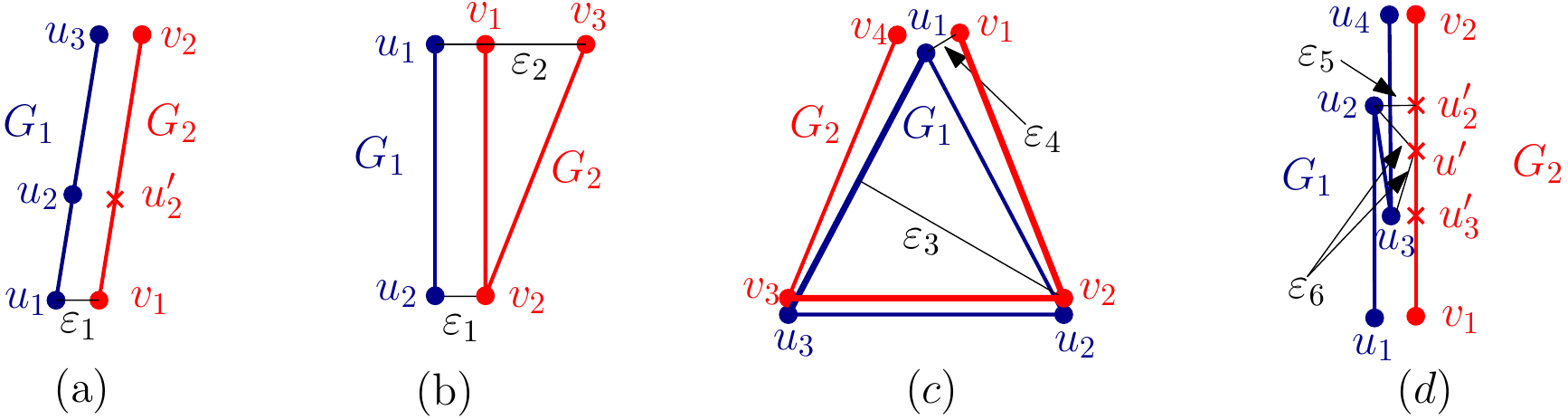}
\caption{Examples of graph mappings $s_1: G_1 \rightarrow G_2$ and
  $s_2: G_2 \rightarrow G_1$, and the resulting graph
  distances. Mapped vertices are drawn with crosses and are not graph
  vertices.\\
  (a) $\distDir(G_1,G_2)=\distDir(G_2,G_1)= \eps_1$.
  $s_1(u_1)=v_1, s_1(u_2)=u'_2, s_1(u_3)=v_2$ and $s_2=s^{-1}_1$.\\
  %
  (b) $\distDir(G_1,G_2)=\eps_1 < \eps_2=\distDir(G_2,G_1)$.
  The mapping $s_1(u_1)=v_1$ and $s_1(u_2)=v_2$ is not surjective, and
  $s_2(v_1)=s_2(v_3)=u_1$ and $s_2(v_2)=u_2$ is not injective.\\
  (c) $\distDir(G_1,G_2)=\eps_3 > \eps_4=\distDir(G_2,G_1)$.
  $s_1(u_i)=v_i$ and $s_2(v_i)=u_i$ for $i=1,2,3$; $s_2(v_4)=u_1$.\\
  (a)-(c) The weak graph distances equal the strong graph distances.\\
  (d) $\distDir(G_1,G_2)=\distDirWeak(G_1,G_2)=\distDirWeak(G_2,G_1)=\eps_5
  < \eps_6=\distDir(G_2,G_1)$. Here, the mappings that attain the
  strong graph distances are $s_1(u_1)=v_1, s_1(u_2)=u'_2,
  s_1(u_3)=u'_3, s_1(u_4)=v_2$ and $s_2(v_1)=u_1, s_2(v_2)=u_4$, where
  $s_2$ in the limit maps $u'$ to all points on the edge from $u_2$ to
  $u_3$. The mappings attaining the weak graph distances are
  $s^w_1=s_1$ and $s^w_2=s^{-1}_1$.  }
  \label{fig:examples}
\end{figure}

In the following, we
show that the traversal distance between a graph $G_1$ and a graph
$G_2$ is a lower bound for $\distDirWeak(G_1,G_2)$, which follows from
the observation that the traversal distance captures the combinatorial
structure of the graphs to a lesser extent than our graph distances.
Furthermore, we apply the graph distances to measure the similarity
between two polygonal paths to examine how these new definitions are
generalizations of the (weak) \fr for curves to graphs.

\subsection{Relation to Traversal Distance}\label{subsec:travdist}
A related distance measure for graphs was proposed by Alt et.~al.~\cite{alt2003262}.
They define the \emph{traversal distance} of two connected embedded graphs $G_1,G_2$ as
\begin{ceqn}
\begin{align*}
\delta_T(G_1,G_2)\ =\ \inf_{f,g}\ \max_{t\in [0,1]}\ ||f(t)-g(t)||\;,
\end{align*}
\end{ceqn}
where $f$ ranges over all traversals of $G_1$ and $g$ over all partial traversals of $G_2$.
A \emph{traversal} of $G_1$ is a continuous, surjective map $f\colon [0,1]\rightarrow G_1$,
and a \emph{partial traversal} of $G_2$ is a continuous map $g\colon [0,1]\rightarrow G_2$.

Thus, graphs $G_1,G_2$ have small traversal distance if there is a traversal of $G_1$ and
a partial traversal of $G_2$ that stay close together. This could also be used for comparing
a graph $G_1$ to a larger graph $G_2$. However, as we observe below,
the traversals do not require to maintain the combinatorial structure of $G_1$ within $G_2$.
First, we observe that our distance measures are stronger distances in the sense that
\begin{ceqn}
\begin{align*}
\delta_T(G_1,G_2)\ \leq\ \distDirWeak(G_1,G_2)\ \leq\ \distDir (G_1,G_2).
\end{align*}
\end{ceqn}

This follows because a graph mapping that realizes $\distDir (G_1,G_2)
\leq \eps$ maps any traversal of $G_1$ to a partial traversal of $G_2$
with distance at most $\eps$. For the weak graph distance, the
traversal might need to be adjusted, so that it moves back along an
already traversed path where the weak Fr\'echet matching requires
it. Note that a traversal need not be injective.

However, the traversal distance captures the combinatorial structure
of the graphs to a lesser extent than our measures.
Figure~\ref{fig:examples}~(c) shows two graphs that have large graph
distance (in particular the directed distance from $G_1$ to $G_2$ is
large) but small traversal distance. If indeed $G_1$ is a map
reconstruction and $G_2$ the ground truth we are comparing to, then
the distance from $G_1$ to $G_2$ should be large.

\subsection{Graph Distance for Paths}
Consider the simple case that the graphs are paths embedded as
polygonal curves.  In this case, the (weak) graph distance is closely
related to the (weak) \fr.  If the graphs are single edges embedded as
polygonal curves, the graph and curve distances are in fact equal
except for orientation of the curves.  If the graphs are paths, such
that each edge is embedded as a straight segment, the curve and graph
distances are related, but not identical, as we show next.

A graph mapping between two paths maps vertices from one path to
points on the other path, and it maps edges to the corresponding
subpaths. In this case, we can characterize the graph distance in the
\emph{free space}~\cite{alt1995computing}, the geometric structure
used for computing the \fr. Recall that for curves $f,g\colon [0,1]
\rightarrow \mathbb{R}^d$ the free space is defined as $ F_\eps(f,g) =
\{ (s,t)\, |\, d(f(s),g(t)) \leq \eps \}$, i.e., the subset of the
product of parameter spaces such that the corresponding points in the
image space have distance at most~$\eps$.

\begin{observation}
Let $P_1, P_2$ be two polygonal curves parameterized over $[0,m]$ and
$[0,n]$, respectively. A graph mapping realizing $\distDir(P_1,P_2)
\leq \eps$ can be characterized as an $x$-monotone path in $[0,m]
\times [0,n]$ from the left boundary to the right boundary that is
$y$-monotone (either increasing or decreasing) in each column of the
free space. A graph mapping realizing $\distDirWeak(P_1,P_2) \leq
\eps$ is characterized by a path in $[0,m] \times [0,n]$ from the left
boundary to the right boundary that is vertex-$x$-monotone, i.e., it
is monotone in the traversal of the vertices on the $x$-axis.
\end{observation}

This observation implies relationships between graph distance and
(weak) \fr that are summarized in Lemma~\ref{lem-comparisonFrechet}.
In this lemma we use the non-standard variant of (weak) \fr that does
not require the homeomorphism to be orientation preserving, but allows
to choose an orientation.
\footnote{For the \wfr one can drop the requirement
that the reparameterizations $\alpha,\beta$ keep the endpoints fixed,
also called \emph{boundary restriction}~\cite{bbkrw-wtd-07}.}
Our graph distance does this naturally by
choosing where to map. This variant of the \fr for curves can be
computed by running the standard algorithm twice, i.e. searching for a
path from bottom-left to top-right corner, as well as from top-left to
bottom-right corner in the free space. Alternatively, we could enforce
an orientation for the graph distance, e.g., using directions on the
graphs.

\begin{lemma}\label{lem-comparisonFrechet}
Let $P_1, P_2$ be paths embedded as polygonal curves. Then
\begin{ceqn}
\begin{align*}
\delta_{wF*}(P_1, P_2) \leq \distWeak(P_1,P_2) \leq \dist(P_1,P_2) \leq \delta_{F}(P_1,P_2),
\end{align*}
\end{ceqn}
where $\delta_{wF*}$ denotes the weak \fr without boundary restriction.

If $P_1, P_2$ are single edges embedded as polygonal paths, equality
holds for both the weak and the strong distances.
\end{lemma}

\begin{proof}
The last inequality holds because a path in the free space realizing
the \fr\ is a monotone path in $x$ and $y$ from the lower left to
upper right corner, hence also realizes both undirected graph
distances.  For the first inequality, observe that two paths in the
free space realizing the two directed weak graph distances can be
combined to a path from the left to the right boundary realizing the
\wfr without boundary restriction.

If $P_1,P_2$ are single edges embedded as polygonal paths, this is
essentially the same as a parameterization, hence both distances are
equal (now with boundary restriction).
\end{proof}
Note that if we require mapping endpoints to themselves then the weak
graph distance is also lower bounded by the \wfr with boundary
restriction.  A 1D-example of two paths where $\delta_{wF}$ is
strictly smaller than $\distWeak$ when enforcing to map endpoints is
the following: $P_1= (0,2,0,2)$ and $P_2= (0,1,2)$. Here,
$\delta_{wF}(P_1, P_2) = 0 < 1 = \distWeak(P_1, P_2)$.

Intuitively, and confirmed by the above lemma, our graph distance
measures are at least as hard to compute as the \fr.  It is known that
the \fr\ cannot be computed in less than subquadratic time unless the
strong exponential time hypothesis
fail~\cite{bringmann2014walking}. Hence we do not expect to compute
our graph distance measures more efficiently than quadratic time.

\section{Algorithms and Hardness for Embedded Graphs}\label{sec:alg}
Throughout this paper, let $G_1=(V_1,E_1)$ and $G_2=(V_2,E_2)$ be two
straight-line embedded graphs, and let $n_1 = |V_1|$, $m_1 = |E_1|$,
$n_2 = |V_2|$ and $m_2 = |E_2|$.

First, we consider the decision variants for the different graph
distances defined in \defref{def:distances}. Given $G_1$ and $G_2$ and
a value $\eps>0$, the decision problem for the graph distances is to
determine whether $\distDir(G_1,G_2)\leq\eps$
(resp., $\distDirWeak(G_1,G_2)\leq\eps$). Equivalently, this amounts
to determining whether there exists a graph mapping from $G_1$ to
$G_2$ realizing $\distDir(G_1,G_2)\leq\eps$ (resp.,
$\distDirWeak(G_1,G_2)\leq\eps$).  Note that the undirected distances
can be decided by answering two directed distance decision
problems. As we show in \secref{sec:opt}, the value of $\eps$ can be
optimized by parametric search.

In \secref{sec:algorithm} we describe a general algorithmic approach
for solving the decision problems by computing {\em valid
  $\eps$-placements} for vertices.  We show that for general embedded
graphs the decision problems for the strong and weak directed graph
distances are NP-hard, see \secref{sec:NPC-general}.  However, we
prove in \secref{sec:algosCases} that our algorithmic approach yields
polynomial-time algorithms for the strong graph distance if $G_1$ is a
tree, and for the weak graph distance if $G_1$ is a tree or if both
are plane graphs. In the latter scenario ($G_1$ and $G_2$ plane
graphs), deciding if $\distDir(G_1, G_2) \leq \eps$ remains NP-hard,
see \secref{sec:NPC-plane}.

\subsection{Algorithmic Approach}
\label{sec:algorithm}
Recall, that a (directional) graph mapping that realizes a given distance $\eps$
maps each vertex of $G_1$ to a point in $G_2$ and each edge of $G_1$
to a simple path in $G_2$ within this distance.  In order to determine
whether such a graph mapping exists, we define the notion of {\em
  $\eps$-placements} of vertices and edges; see
Figures~\ref{fig:edge-placement} and~\ref{fig:valid-placement}~(a).

\begin{definition}[$\eps$-Placement]
An \emph{$\eps$-placement} of a vertex $v$ is a maximally connected
part of $G_2$ restricted to the $\eps$-ball $B_\eps(v)$ around $v$.
An \emph{$\eps$-placement} of an edge $e=\{u,v\} \in E_1$ is a path
$P$ in $G_2$ connecting placements of $u$ and $v$ such that
$\delta_F(e,P)\leq \eps$.
In that case, we say that $C_u$ and $C_v$ are {\em reachable} from
each other.  An \emph{$\eps$-placement} of $G_1$ is a graph mapping
$s\colon G_1\rightarrow G_2$ such that $s$ maps each edge $e$ of $G_1$
to an $\eps$-placement.

A \emph{weak $\eps$-placement} of an edge $e=\{u,v\}$ is a path $P$ in
$G_2$ connecting placements of $u$ and $v$ such that
$\delta_{wF}(e,P)\leq \eps$.  A \emph{weak $\eps$-placement} of $G_1$
is a graph mapping $s\colon G_1\rightarrow G_2$ such that $s$ maps
each edge $e$ of $G_1$ to a weak $\eps$-placement.
\label{def:eps-placements}
\end{definition}

\begin{figure}[b]
	\subfloat[An $\eps$-placement of $e$.]{
		\includegraphics[page=6,scale=0.6]{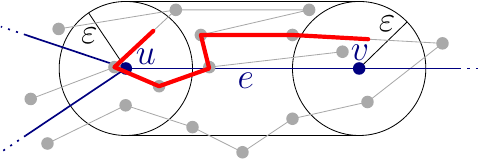}
		\label{subfig:edge-placement}}
	\hfill
	\subfloat[Not an $\eps$-placement.]{
		\includegraphics[page=7,scale=0.6]{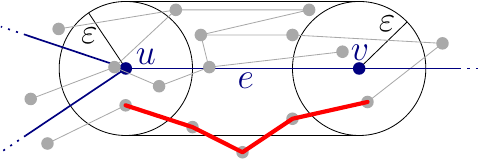}
		\label{subfig:not-edge-placement}}
	\hfill
	\subfloat[A weak $\eps$-placement.]{
		\includegraphics[page=8,scale=0.6]{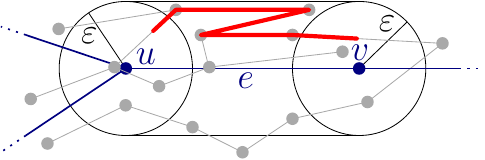}
		\label{subfig:weak-edge-placement}}
	\caption{(a) Illustration of $\eps$-placements of an edge $e$. (b) Not an $\eps$-placement because the path leaves the $\eps$-tube around $e$. (c) The \fr is too large, but $e$ can be mapped to the path if backtracking is allowed. Thus, it is a weak $\eps$-placement.}
	\label{fig:edge-placement}
\end{figure}

\begin{figure}
\centering
\includegraphics[page=1,scale=0.7]{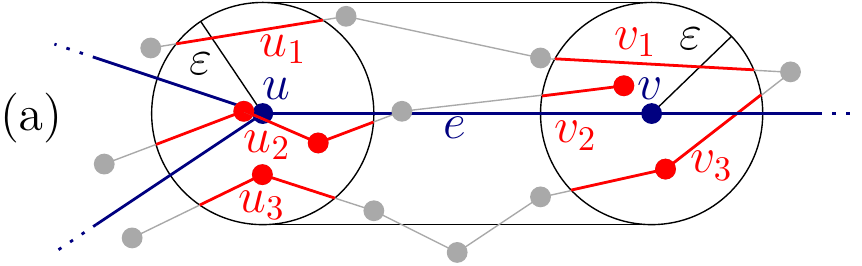}
\hspace{1cm}
\includegraphics[page=4,scale=0.7]{valid-placements}
\caption{Illustration of valid and invalid vertex placements. (a)
  Placements $u_3$ (resp.~$v_3$) are invalid because they are not
  connected to a placement of $v$ (resp.~$u$) by an $\eps$-placement
  of the edge $e$. Placement $v_2$ is valid when considering $e$ in
  isolation, but it cannot connect to a placement for the edge that
  leaves $v$ to the right. Thus, it is also invalid. As a result of
  pruning $v_2$ (right), $u_2$ becomes invalid as well, leaving only
  $u_1$ and $v_1$ as potentially valid placements of $u$ and $v$ (b).}
\label{fig:valid-placement}
\end{figure}

Note that an $\eps$-placement of a vertex $v$ consists of edges and
portions of edges of $G_2$, depending whether $B_\eps(v)$ contains
both, one or zero endpoint(s) of the edge, see
\figref{fig:valid-placement}.
Also note that each vertex has $O(m_2)$ $\eps$-placements, 
since an $\eps$-placement is defined as a connected part of $G_2$ of maximal size inside $B_\eps(v)$.
Furthermore, we consider two graph
mappings $s_1$ and $s_2$ from $G_1$ to $G_2$ to be equivalent in terms
of the directed (weak) graph distance if for each vertex $v \in V_1$,
$s_1(v)$ and $s_2(v)$ are points on the same $\eps$-placement of $v$.

\paragraph*{General Decision Algorithm.}
Our algorithm consists of the following four steps, which we describe
in more detail below. We assume $\eps$ is fixed and use the term
\emph{placement} for an $\eps$-placement.

Observe that each connected component of $G_1$ needs to be mapped to a
connected component of $G_2$, and each connected component of $G_1$
can be mapped independently of the other components of $G_1$. Hence we
can first determine the connected components of both graphs, and then
consider mappings between connected components only. In the following
we present an algorithm for determining if a mapping from $G_1$ to
$G_2$, that realizes a given distance $\eps$, exists, where both $G_1$
and $G_2$ are connected graphs.
\begin{algorithm}
  \caption{General Decision Algorithm}
  \begin{algorithmic}[1]
    \State Compute vertex placements.
    \State Compute reachability information for vertex placements.
    \State Prune invalid placements.
    \State Decide if there exists a placement for the whole graph $G_1$.
  \end{algorithmic}
  \label{alg:general}
\end{algorithm}

\paragraph*{1. Compute vertex placements.}
We iterate over all vertices $v\in V_1$ and compute all their
placements.  Each vertex has $O(m_2)$ placements, so the total number
of vertex-placements is $O(n_1 \cdot m_2)$, and they can be computed
in $O(n_1 \cdot m_2)$ time using standard algorithms for computing
connected components.

\paragraph*{2. Compute reachability information of vertex placements.}
Next, we iterate over all edges $e=\{u,v\}\in E_1$ to determine all
placements of its vertices that allow a placement of the edge.  That
is, we search for all pairs of vertex-placements $C_u,C_v$ that are
reachable from each other according to Definition
\ref{def:eps-placements}.

For the weak graph distance, we need to find all pairs of placements
of $u$ and placements of $v$ that can reach one another using paths
contained in the $\eps$-tube $T_\eps(e)$ around $e$, i.e., the set of
all points with distance $\leq\eps$ to a point on $e$, see
\figref{fig:edge-placement}~(c).  If we restrict $G_2$ to its
intersection with the $\eps$-tube, all placements in the same
connected component are mutually reachable.  Thus, each edge is
processed in time linear in the size of $G_2$ using linear space per
edge: For each connected component a pair of lists containing the
placements of $u$ and $v$ in that component, respectively, is
computed.  So, all reachability information can be computed in $O(m_1
\cdot m_2)$ time and space.
Note that the weak \fr between a straight line edge $e \in E_1$ and a
simple path $s(e)$ in $G_2$ is the maximum of the Hausdorff distance
between $e$ and $s(e)$ and the distances of the endpoints of $e$ and
$s(e)$.

For the strong graph distance, existence of a path inside the
$\eps$-tube is not sufficient to describe the connectivity between
placements.  We must ensure that the Fr\'echet distance between $e$
and $P$ is at most $\eps$, i.e.,\ a continuous and monotone map $s$
must exist from $e$ to $P$ such that $\delta_F(t,s(t))\leq\eps$ for
all $t\in e$.  This can be decided in $O(|P|)$ time using the original
dynamic programming algorithm for computing the Fr\'echet
distance~\cite{alt1995computing}.  In order to determine whether such
a path $P$ exists, every placement of $u$ stores a list of all
placements of $v$ that are reachable.  The connectivity information
can be computed by running a graph exploration, starting from each
placement, which prunes a branch if the search leaves the $\eps$-tube
or backtracking on $e$ is required to map it.  This method runs a
search for every placement of the start vertex and thus needs
$O(m_2^2)$ time per edge of $G_1$.  Since the connectivity is
explicitly stored as pairs of placements that are mutually reachable,
it also needs $O(m_2^2)$ space per edge. Hence, in total over all
edges, $O(m_1 \cdot m_2^2)$ time and space are needed.  Summing up, we
have:
\begin{lemma}
To run step~1 and step~2 of \algoref{alg:general}, we need $O(m_1
\cdot m_2)$ time and space for the weak graph distance and $O(m_1
\cdot m_2^2)$ time and space for the strong graph distance.
\label{lem:runtime_reachability}
\end{lemma}

\paragraph*{3. Prune invalid placements.}
Now, after having processed all vertices and edges,
it still needs to be decided whether $G_1$ as a whole can be
mapped to $G_2$.  To this end, we delete \emph{invalid} placements of
vertices.

\begin{definition}[Valid Placement]
An $\eps$-placement $C_v$ of a vertex $v$ is \emph{(weakly) valid} if
for every neighbor $u$ of $v$ there exists an $\eps$-placement $C_u$
of $u$ such that $C_v$ and $C_u$ are connected by a (weak)
$\eps$-placement of the edge $\{u,v\}$. Otherwise, $C_v$ is
\emph{(weakly) invalid}.
\end{definition}

See Figure~\ref{fig:valid-placement} for an illustration of (in)valid
placements. As shown in the Figure, deleting an invalid placement
possibly sets former valid placements to be invalid. Thus, we need to
process all placements recursively until all invalid placements are
deleted and no new invalid placements occur. Note that the ordering of
processing the placements does not affect the final result.
To decide which placements of vertices $u$ and $v$ incident to an
edge~$e$ are valid, we use the reachability information computed in
Step 2.

Initially there are $O(n_1 \cdot m_2)$ vertex-placements, each of
which may be deleted once.  For the weak graph distance, connectivity
is stored using connected components inside the $\eps$-tube
surrounding an edge $\{u,v\}$.  On deleting a placement $C_v$ of $v$,
it is removed from the list containing placements of $v$.  If a
component no longer contains placements of $v$ (i.e. its list becomes
empty), then all placements of $u$ in that component become invalid.
A placement $C_v$ is deleted at most once and upon deletion it must be
removed from one list for every edge incident to $v$.  Thus, the time
for pruning $C_v$ is $O(\deg(v))$.  Since the sum of all degrees is
$2m_1$, all invalid placements can be pruned in $O(m_1 \cdot m_2)$
time.
For the strong graph distance, every placement has a list of
placements to which it is connected.  On deleting $C_v$, it must be
removed from the lists of all placements $C_u$ to which $C_v$ is
connected.  Each vertex has $O(m_2)$ placements which have to be
removed from a list for each neighbor of $v$.  Thus, pruning a
placement runs in $O(\deg(v)\cdot m_2)$ time and pruning all invalid
placements in $O(m_1 \cdot m_2^2)$ time.

\begin{lemma}
Pruning all invalid placements takes $O(m_1 \cdot m_2)$ time for the
weak graph distance and $O(m_1 \cdot m_2^2)$ time for the strong graph
distance.
\label{lem:runtime_pruning}
\end{lemma}

Note that after the pruning step all remaining vertex placements are
(weakly) valid. However, the existence of a (weakly) valid placement
for each vertex is not a sufficient criterion for $\distDir(G_1, G_2)$
($\distDirWeak(G_1, G_2)$) in general, see
\figref{fig:plane-valid-counterexample}.

\paragraph*{4. Decide if there exists a placement for the whole graph $G_1$.}
After pruning all invalid placements, we want to decide if the
remaining valid vertex-placements allow a placement of the whole graph
$G_1$.
The complexity of this step depends on the graph and the distance measure: 
for plane graphs we show that we 
can concatenate weakly valid placements of two adjacent faces
(\lemref{lem:gd-weak-plane}), whereas this is not possible for the
directed strong graph distance in this setting
(\thmref{theorem:NP_plane}) or for general graphs for both
distances
(\thmref{thm:NP_hard_general}).
Although deciding the directed (weak) graph distance is NP-hard for
general graphs, there are two settings which may occur after running
steps 1-3 of Algorithm~\ref{alg:general}, making step 4 of the
algorithm trivial.  Clearly $\distDir(G_1,G_2)>\eps$
($\distDirWeak(G_1,G_2)>\eps$) if there is a vertex that has no
(weakly) valid $\eps$-placement. Furthermore, we have the
following:

\begin{lemma}
  \label{lem:unique-valid-pl}
If, after running steps 1-3 of Algorithm \ref{alg:general}, each
internal vertex (degree at least two) has exactly one valid
$\eps$-placement (resp., weakly valid $\eps$-placement) and each
vertex of degree one has at least one valid $\eps$-placement (resp.,
weakly valid $\eps$-placement), then $G_1$ has an $\eps$-placement
(resp., weak $\eps$-placement). Thus, $\distDir(G_1,G_2) \leq \eps$
(resp., $\distDirWeak(G_1,G_2) \leq \eps$).
\end{lemma}

\lemref{lem:runtime_reachability}, \lemref{lem:runtime_pruning} and
\lemref{lem:unique-valid-pl} imply the following Theorem.

\begin{theorem}
\label{thm:unique-valid-pl}
If there is a vertex that has no valid $\eps$-placement or if each
vertex has exactly one valid $\eps$-placement after running steps~1-3
of Algorithm \ref{alg:general}, the directed strong graph distance can
be decided in $O(m_1 \cdot m_2^2)$ time and space.  Analogously, if
there is a vertex that has no weakly valid $\eps$-placement or if each
vertex has exactly one weakly valid $\eps$-placement after running
steps~1-3 of Algorithm \ref{alg:general}, the directed weak graph
distance can be decided in $O(m_1 \cdot m_2)$ time and space.
\end{theorem}
\begin{proof}
Map each internal vertex $v$ to a point $s(v)$ on 
of its unique (weakly) valid placement $C_v$. 
Consider an edge $e=\{u,v\}\in E_1$.  In the previous step, at least
one (weak) placement $P_e$ of $e$ was discovered that connects points
$p_0$ and $p_k$ on $C_u$ and $C_v$, respectively, since otherwise
$C_u$ and $C_v$ would be invalid.  If $p_k \neq s(v)$, $P_e$ can be
adapted by shortening it and/or concatenating a path on $C_v$
(i.e.\ inside $B_\eps(v)$) without causing its (weak) \fr to $e$ to
become $>\eps$.  Adapt $P_e$ to a path $P'_e$ that has endpoints
$p'_0=s(u)$ and $p'_k=s(v)$ and define $s(e)=P'_e$.  Now, $s$ is a
graph mapping from $G_1$ to $G_2$ and each edge is mapped to a path
with (weak) \fr at most $\eps$, so $\distDir(G_1,G_2)\leq\eps$ (or
$\distDirWeak(G_1,G_2)\leq\eps$). Recall, that each vertex $w$ of
degree one is either connected to an internal vertex $i$, or $G_1$
consists of only one edge $\{w,x\}$. The first case is already
covered, since the unique valid (weak) placement $C_i$ for $i$ is
reachable from any valid (weak) placement of $w$. The latter case
follows because every vertex placement for $w$ is valid, i.e., for
each placement $C_w$ for $w$ there is a (weakly) reachable placement
$C_x$ for $x$ and any combination of two reachable placements $C_w$
and $C_x$ yields a valid (weak) placement of $G_1$.
\end{proof}

\subsection{NP-Hardness for the General Case}
\label{sec:NPC-general}
Notwithstanding the special cases in \thmref{thm:unique-valid-pl}, deciding the (weak) graph distance is not tractable for general graphs.
\begin{theorem}
Deciding whether $\distDir(G_1,G_2) \leq\eps$ and deciding whether $\distDirWeak(G_1,G_2)\leq\eps$ for two graphs $G_1$ and $G_2$ embedded in $\mathbb{R}^2$ is NP-hard.
\label{thm:NP_hard_general}
\end{theorem}
\begin{proof}
We show NP-hardness with a reduction from binary constraint
satisfaction problem (CSP), which is defined as follows:

\begin{problem}
\proc{Binary Constraint Satisfaction Problem (CSP)}\\
\noindent \textbf{Instance:} A set of variables
$X=\{x_1,\ldots,x_n\}$, variable domains \\ $D = \{D_1,\ldots,D_n\}$
and a set of constraints $C = \{C_1,\ldots,C_m\}$, where each
constraint $C$ has two variables $x_i$, $x_j$ and a relation $R_C
\subseteq D_i\times D_j$.\\
\noindent \textbf{Question:} Can each variable $x_i$ be assigned a
value $d_i\in D_i$ such that for each constraint $C$ on variables
$x_i$, $x_j$, their values $(d_i,d_j)$ satisfy $R_C$?
\end{problem}

Consider an instance $\langle X, D, C\rangle$.  Set $\eps=1$.  Every
variable $x_i$ is represented by a vertex $v_i$ in $G_1$ and for each
constraint $C_k$ on variables $x_i$, $x_j$, $G_1$ has an edge
$\{v_i,v_j\}$ in $G_1$.  We embed $G_1$ such that all adjacent
$\eps$-balls are separated by at least $2\eps$ and no $\eps$-tube of
an edge overlaps an $\eps$-ball that does not belong to one of the
edge's endpoints.  This can for example be realized by placing all
vertices on a sufficiently large circle.

Every value $d_{i,a}\in D_i$ is represented by a vertex $u_{i,a}$ in
$G_2$ that is inside the $\eps$-ball of $v_i$.  For each pair of
values $d_{i,a}\in D_i$, $d_{j,b}\in D_j$ allowed by a constraint on
$x_i$ and $x_j$, $G_2$ has an edge $\{u_{i,a},u_{j,b}\}$.  This way,
every vertex $u_{i,a}$ of $G_2$ defines exactly one $\eps$-placement
of the corresponding vertex $v_i$ in $G_1$.
Figure~\ref{fig:csp-reduction} illustrates the construction.
\begin{figure}[t]
\centering
\includegraphics[width=.5\textwidth]{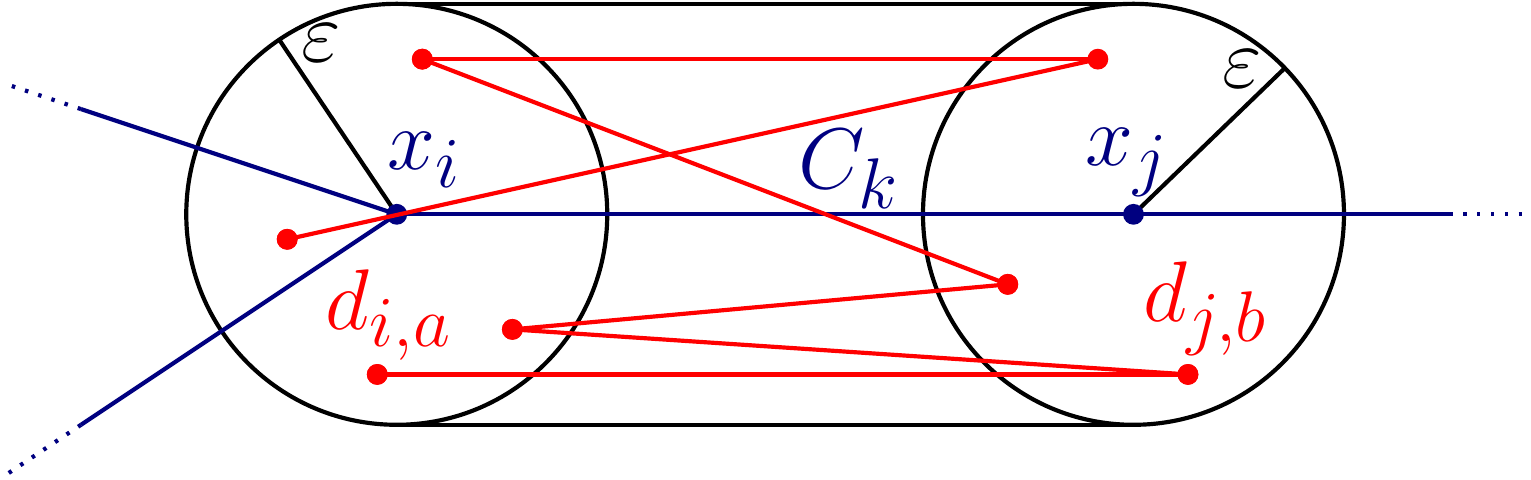}
\caption{Illustration of the reduction from BCSP.}
\label{fig:csp-reduction}
\end{figure}

A solution to the CSP consists of selecting a value $d_i\in D_i$ for
each variable $x_i$, such that if there is a constraint on $v_i$ and
$v_j$, the pair $\{d_i,d_j\}$ satisfies the constraint.  This is
equivalent to selecting a placement $d_i$ of each vertex $v_i$, such
that if $G_1$ has an edge $\{v_i,v_j\}$, then $G_2$ has an edge
connecting $u_i$, $u_j$ representing $d_i$ and $d_j$, respectively.
The graph distance problem has the weaker requirement that there
exists a path $P_{i,j}$ between $u_i$ and $u_j$, such that
$\delta_F(\{v_i,v_j\},P_{i,j})\leq\eps$.  However, the construction is
such that only paths consisting of a single edge are permitted for the
strong distance, since $\eps$-balls must be sufficiently separated and
nonoverlapping with $\eps$-tubes.  So $G_2$ must have an edge
$\{u_i,u_j\}$ if $G_1$ has an edge $\{v_i,v_j\}$.  So,
$\distDir(G_1,G_2)\leq\eps$ if and only if $\langle X,D,C \rangle$ is
a satisfiable binary CSP.

For the weak graph distance, edges in $G_1$ can be mapped to paths
consisting of multiple edges in $G_2$.  In this case, there may be
weak placements of $G_1$ that do not represent a solution to the
constraint satisfaction instance.  To remedy this, we insert a vertex
in the middle of each edge of $G_1$.  The vertex is placed such that
its $\eps$-ball is separated from the $\eps$-balls of the original
endpoints of the edge by at least $2\eps$, so each of the new edges is
mapped to part of a single edge in $G_2$.  In this construction
$\distDirWeak(G_1,G_2)\leq\eps$ if and only if $\langle X,D,C \rangle$
is a satisfiable binary CSP.
\end{proof}

\subsection{Efficient Algorithms for Plane Graphs and Trees}
\label{sec:algosCases}

Here, we show that that \algoref{alg:general} yields polynomial-time
algorithms for deciding the strong graph distance if $G_1$ is a tree
(\thmref{thm:gd-strong-dist}), and the weak graph distance if $G_1$ is
a tree or if both are plane graphs (\thmref{thm:gd-weak-dist}). More
precisely, we show that the existence of at least one (weakly) valid
placement for each vertex is a sufficient condition for
$\distDir(G_1,G_2)\leq\eps$ or $\distDirWeak(G_1,G_2)\leq\eps$.

\begin{lemma}
\label{lem:gd-map-tree}
If $G_1$ is a tree and every vertex of $G_1$ has at least one (weakly)
valid $\eps$-placement after running steps 1-3 of Algorithm
\ref{alg:general}, then $G_1$ has a (weak) $\eps$-placement. Thus,
$\distDir(G_1,G_2)\leq\eps$ (or $\distDirWeak(G_1,G_2)\leq\eps$).
\end{lemma}
\begin{proof} 
We view $G_1$ as a rooted tree, selecting an arbitrary vertex as the
root.  We map all vertices of $G_1$ from the root outwards.  First,
map the root to an arbitrary (weakly) valid placement.  When
processing a vertex $v$: Map $v$ to an arbitrary (weakly) valid
placement that is reachable from the placement its parent $p$ is
mapped to.  Recall, that a (weak) $\eps$-placement of a vertex is
\emph{(weakly) valid} if there is an $\eps$-placement for every
incident edge. Since $p$ was mapped to a (weakly) valid placement and
there is an edge $\{p,v\}$ in $G_1$, there must be at least one such
placement of $v$ by definition of a (weakly) valid placement.  Since
all edges in $G_1$ are tree edges, this ensures that every edge is
mapped correctly, that is to a path with (weak) \fr\ at most $\eps$.
\end{proof}


\begin{lemma}
\label{lem:gd-weak-plane}
If $G_1$ and $G_2$ are plane graphs and every vertex of $G_1$ has at
least one weakly valid $\eps$-placement after running steps 1-3 of
Algorithm \ref{alg:general}, then $G_1$ has a weak
$\eps$-placement. Thus, $\distDirWeak(G_1,G_2)\leq\eps$.
\end{lemma}
\begin{proof}
A \emph{tree-substructure} of $G_1$ is a tree $T = (V_T, E_T)$ induced
by the vertex set $V_T \subset V_1$ with a root vertex $r \in V_T$,
such that for all vertices $v\in V_T$, $v \neq r$, $v$ is not an
endpoint of an edge $e \in E_1 \backslash E_T$ and such that $T$ is
maximal, in the sense that when adding one additional vertex, $T$
contains a cycle. We first remove all tree-substructures of $G_1$ and
map these as in the proof of \lemref{lem:gd-map-tree}.  Next, we
consider all faces of the remainder of $G_1$ and show how to
iteratively map them.

Consider a cycle $C$ bounding a face $F$ and let $e_1$ and $e_2$ be
two edges of $C$ incident to a vertex $v$. Let $b$ be the line segment
of the bisector of $e_1$ and $e_2$ inside $B_\eps(v)$. We define the
\emph{outermost placement} of $v$ as the placement which intersects
$b$ at maximum distance to the endpoint of $b$ inside $F$, see
Figure~\ref{fig:gd-outer-placement}~(a).  Furthermore, we define an
\emph{outermost path} in $G_2$ of an edge $e=\{u,v\}$ of $G_1$ as the
path $P_{out}$ with maximum distance to $F$ connecting the outermost
placements of $u$ and $v$. That is, no subpath $Q$ of $P_{out}$ can be
replaced by a path $R$ such that $\delta_H(R,B) \leq \delta_H(Q,B)$,
where $\delta_H$ is the Hausdorff distance and $B$ is the boundary of
the tube $T_\eps(e)$ which lies inside the face $F$.  Note that if an
edge is shorter than $2\eps$, and hence the $\eps$-balls around the
vertices overlap, then so possibly do the placements. In particular,
in this case the outer placements may overlap, in which case the edge
placements degenerate, see Figure~\ref{fig:gd-outer-placement}~(a).
Finally, we define an \emph{outer placement} $O$ of $C$ in $G_2$ as
the concatenation of all outermost paths of edges of $C$.

Note that if $C$ is sufficiently convex the outer placement is simply
the cycle that bounds $H$.  See
Figure~\ref{fig:gd-outer-placement}~(b) for an example, where the red
outer placement bounds the outer face of $G_2$ restricted to red and
pink vertices and edges.  The outer placement of $C$ is a weak
$\eps$-placement of $C$.

\begin{figure}
\centering
\subfloat[A vertex with its placement.]{
	\includegraphics[page=1, scale=0.6]{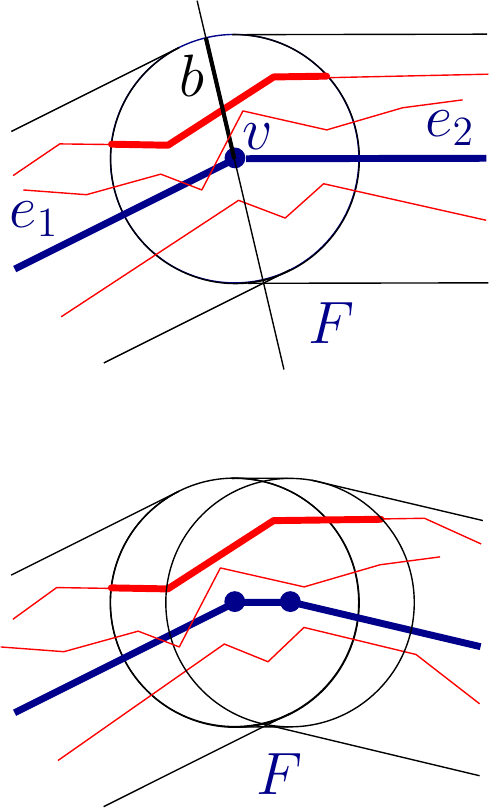}
	\label{subfig:outer-placement}}
\hfill
\subfloat[A cycle with its outer placement.]{
	\includegraphics[page=1, scale=1.3]{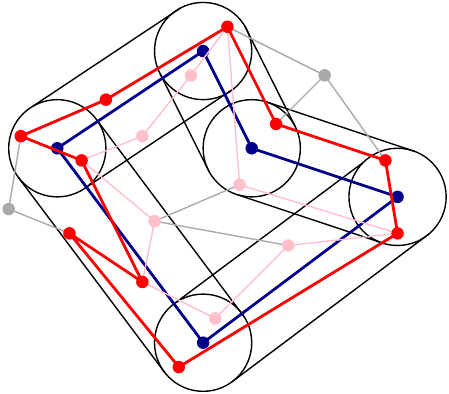}
	\label{subfig:gd-outer-placement}}
\hfill
\subfloat[Merging two outer placements.]{
	\includegraphics[page=2, scale=1.15]{weak-dist-plane}
	\label{subfig:gd-merge-outer-placements}}
\caption{Illustration of outer placements and how to merge them. In (c) 
the outer placements of cycles $C_1$ and $C_2$ can be merged by mapping the shared path $P$ through~$o_1$.}
\label{fig:gd-outer-placement}
\end{figure}

Now, consider two cycles $C_1$ and $C_2$ bounding adjacent faces of
$G_1$, which share a single (possibly degenerate) path $P$ between
vertices $u$ and $v$.  Let $O_1$ and $O_2$ be the outer placements of
$C_1$ and $C_2$, respectively.  By definition of an outermost
placement, $O_1$ and $O_2$ must intersect inside the intersection of
the $\eps$-tubes of $C_1$ and $C_2$. Let $o_1$ and $o_2$ of $O_1$ and
$O_2$ be the parts between the intersections of $O_1$ and $O_2$
containing the respective images of $P$. Again, by definition of an
outermost placement, it holds that $o_1$ is completely inside $O_2$
and $o_2$ is completely inside $O_1$.

This is illustrated in Figure~\ref{fig:gd-outer-placement}~(c). By
planarity there must be a vertex at the intersections of $O_1$ and
$O_2$. Thus, we can construct a mapping $O'_2$ of $C_2$ that consists
of $o_1$ and $O_2\setminus o_2$.  This is a weak $\eps$-placement of
$C_2$ for which the image of the shared path $P$ is identical to its
image in $O_1$.  Thus, we can merge $O_1$ and $O'_2$ to obtain a weak
$\eps$-placement of these two adjacent cycles.  Note that the mapping
of $C_1$ is not modified in this construction.  Additionally, the
image of the cycle bounding the outer face is its outer placement.
The same argument can be applied iteratively when $C_1$ and $C_2$
share multiple paths.

If there are two cycles $C_1$ and $C_2$ which are connected by a path
$P$ such that one endpoint $u$ of $P$ lies on $C_1$, the other
endpoint $v$ of $P$ lies on $C_2$ and all other vertices of $P$ are no
vertices of $C_1$ or $C_2$, we can still construct a common placement
for $C_1$, $C_2$ and $P$: Let $C_u$, $C_v$ be the outermost placements
of $u$ and $v$, respectively and let $D_v$ be a vaild placement of $v$
which is connected by a path $Q$ in $G_2$ to $C_u$ such that
$\delta_{wF}(Q,P) \leq \eps$. Such a placement $D_v$ must exist as
$C_u$ is a valid placement. If $D_v$ = $C_v$ we have found a common
valid placement for $C_1$, $C_2$ and $P$. If $D_v \neq C_v$, by
definition of an outermost placement, the path $Q$ must intersect the
outermost placement $O$ of $C_2$ inside the intersection of the tubes
$T_\eps(P)$ and $T_\eps(C_2)$. As $G_2$ is plane, there is a vertex
$w$ at the intersection and the resulting path $R=Q_{C_u\rightarrow w}
+ O_{w \rightarrow C_v}$ with $\delta_{wF}(R,P)\leq\eps$ connects
$C_u$ and $C_v$.

Now, we iteratively map the cycles bounding faces of $G_1$ until $G_1$
is completely mapped.  Let $\langle F_1,F_2,\ldots,F_k\rangle$ be an
ordering of the faces of $G_1$ such that each $F_i$, for $i\geq 2$ is
on the outer face of the subgraph $\G_{i-1} := C_1\cup
C_2\cup\ldots\cup C_{i-1}$ of $G_1$, where $C_j$ is the cycle bounding
face $F_j$.  Thus, let $F_1$ be an arbitrary face of $G_1$ and
subsequently choose faces adjacent to what has already been mapped.
Hence when adding a cycle $C_i$, we have already mapped $\G_{i-1}$
such that the cycle bounding its outer face is mapped to its outer
placement.  Thus, we can treat $\G_{i-1}$ as a cycle, ignoring the
part of it inside this cycle, and merge its mapping with $C_i$ using
the procedure described above.  This leaves the mapping of $\G_{i-1}$
unchanged, hence this is still a weak $\eps$-placement of $\G_{i-1}$.
However, the mapping of $C_i$ is now modified to be identical to that
of $\G_{i-1}$ in the parts where they overlap.  Thus, we can merge
these mappings to obtain a weak $\eps$-placement of $\G_i$.  After
mapping $F_k$ we have completely mapped $G_1$.
\end{proof}
\lemref{lem:runtime_reachability} and \lemref{lem:runtime_pruning}
together with \lemref{lem:gd-map-tree} and \lemref{lem:gd-weak-plane}
directly imply the following theorems. Note, that $m_1 = O(n_1)$ for
plane graphs and trees, in particular.

\begin{theorem}[Decision Algorithm for Weak Graph Distance]
\label{thm:gd-weak-dist}
Let $\eps>0$. If $G_1$ is a tree, or if $G_1$ and $G_2$ are plane graphs,
then \algoref{alg:general} decides whether $\distDirWeak(G_1,G_2)
\leq\eps$ in $O(n_1 \cdot m_2)$ time and space.
\end{theorem}

\begin{theorem}[Decision Algorithm for Graph Distance]
  \label{thm:gd-strong-dist}
  Let $\eps>0$. If $G_1$ is a tree, then \algoref{alg:general} decides
  whether $\distDir(G_1,G_2) \leq\eps$ in $O(n_1 \cdot m_2^2)$ time
  and space.
\end{theorem}

\paragraph*{Computing the Distance}
\label{sec:opt}
To compute the graph distance, we proceed as for computing the \fr
between two curves: We search over a set of critical values and employ
the decision algorithm in each step. The following types of critical
values can occur:
\begin{enumerate}
\item A new vertex-placement emerges: An edge in $G_2$ is at distance $\eps$ from a vertex in $G_1$. 
\item Two vertex-placements merge: The vertex in $G_2$ where they connect is at distance $\eps$ from a vertex in $G_1$.
\item The (weak) \fr between a path and an edge is $\eps$: these are
  described in~\cite{alt1995computing}. There are exponentially many paths in $G_2$, but each value the Fr\'echet distance may attain is defined by either a vertex and an edge, or two vertices and an edge.
\end{enumerate}
There are $O(n_1\cdot m_2)$ critical values of the first two types,
and $O(m_1\cdot n_2^2)$ of type three.  Parametric search can be used
to find the distance as described in~\cite{alt1995computing}, using
the decision algorithms from Theorems~\ref{thm:gd-weak-dist}
and~\ref{thm:gd-strong-dist}. This leads to a running time of
$O(n_1\cdot m_2 \cdot \log (n_1+n_2))$ for computing the weak graph
distance if $G_1$ is a tree or both are plane graphs. And the total
running time for computing the graph distance if $G_1$ is a tree is
$O(n_1\cdot m_2^2\cdot \log(n_1+n_2))$.

\section{Hardness Results and Algorithms for Plane Graphs}\label{sec:plane_graphs}
Lemma~\ref{lem:gd-weak-plane} does not hold for plane graphs and the
directed strong graph distance because in general outer placements of
cycles cannot be combined to a placement of $G_1$ as shown in the
proof of Lemma~\ref{lem:gd-weak-plane}. See
Figure~\ref{fig:plane-valid-counterexample} for a counterexample. In
fact we show that deciding the directed strong graph distance for
plane graphs is NP-hard.

\begin{figure}
\centering
\includegraphics[width=.7\textwidth]{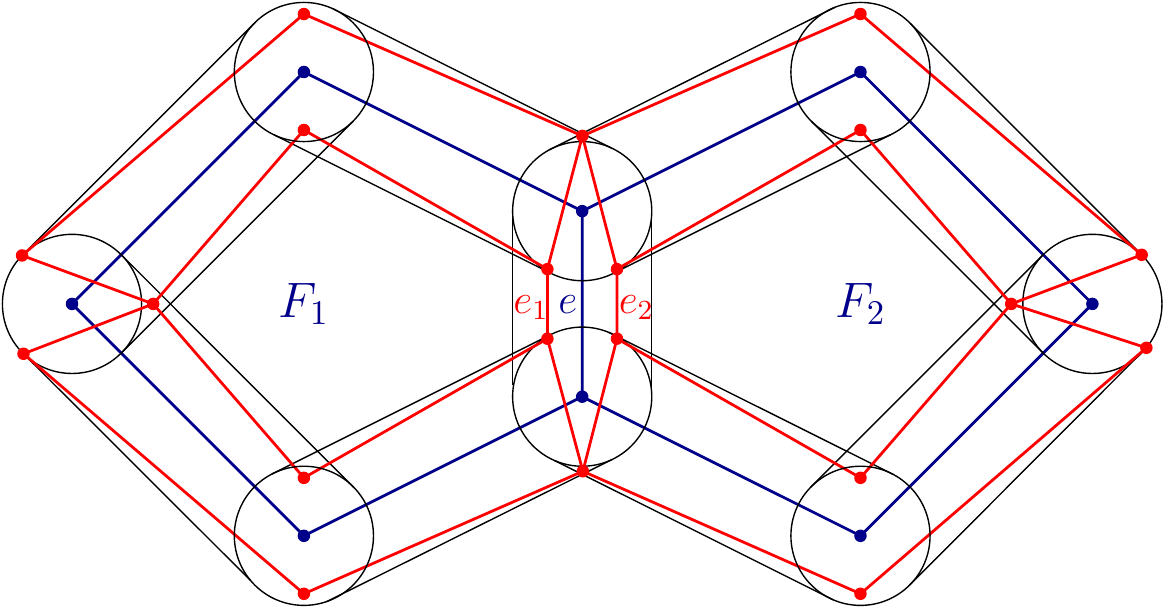}
\caption{An example of plane graphs $G_1$ (blue) and $G_2$ (red) where
  every vertex of $G_1$ has two valid placements, but there is no
  $\eps$-placement of $G_1$: If the central edge $e$ is mapped to a
  path through $e_1$, there is no way to map the cycle bounding face
  $F_2$ on the right, and if $e$ is mapped to a path through $e_2$,
  the cycle bounding $F_1$ cannot be mapped.}
\label{fig:plane-valid-counterexample}
\end{figure}

\subsection{NP-Hardness for the Strong Distance for Plane Graphs}
\label{sec:NPC-plane}
\begin{theorem}
For plane graphs $G_1$ and $G_2$, deciding whether $\distDir(G_1,G_2) \leq\eps$ is NP-hard.
\label{theorem:NP_plane}
\end{theorem}
\begin{proof}
We prove the NP-hardness by a reduction from
\textsc{Monotone-Planar-3-Sat}. In this \textsc{3-Sat} variant, the
associated graph with edges between variables and clauses is planar
and each clause contains only positive or only negative literals. The
overall idea is to construct two graphs $G_1$ and $G_2$ based on a
\textsc{Monotone-Planar-3-Sat} instance $A$, such that $A$ is
satisfiable if and only if $\distDir(G_1,G_2) \leq\eps$. That is, we
construct subgraphs of $G_1$ and $G_2$, where some edges of $G_2$ are
labeled \textsc{True} or \textsc{False} in a way, such that only
certain combinations of \textsc{True} and \textsc{False} values can be
realized by a placement of $G_1$ to $G_2$.  To realize other
combinations, backtracking, at least along one edge of $G_1$, is
necessary. But this is not allowed for the strong graph distance.  In
the following, we describe the construction of the gadgets (subgraphs)
for the variables and the clauses of a \textsc{Monotone-Planar-3-Sat}
instance. Additionally, we need a gadget to split a variable if it is
contained in several clauses and a gadget which connects the variable
gadgets with the clause gadgets. Furthermore, we prove for each gadget
which \textsc{True} and \textsc{False} combinations can be realized
and which combinations are not possible.  All constructed edges are
straight line edges. The graph $G_1$ is shown in blue color and $G_2$
is shown in red color in the sketches used to illustrate the ideas of
the proof. We denote $T_{\eps}(e)$ as the $\eps$-tube around the edge
$e$. All vertices of the graph can be either mapped arbitrarily within
a given $\eps$-surrounding and with a given minimal distance from each
other or must lie at the intersection of two lines. Thus, we can
ensure that the construction uses rational coordinates only and can be
computed in polynomial time.

Following a reparameterization of an edge or a path (embedded curve)
as described in Definition \ref{def:distances}, we denote by the term
"walking along an edge or a path". In this sense backtracking means
that the reparametization is not injective and thus, if backtracking
is necessary to stay within $\eps$-distance to another path or edge
one can conclude that the \fr distance between these paths or edges is
greater than $\eps$ as in this case it is required that the
reparameterization is injective.

Furthermore, we call a path labeled \textsc{True} (\textsc{False})
shortly a \textsc{True} (\textsc{False}) signal.

\begin{figure}
\centering
\includegraphics[scale=0.745]{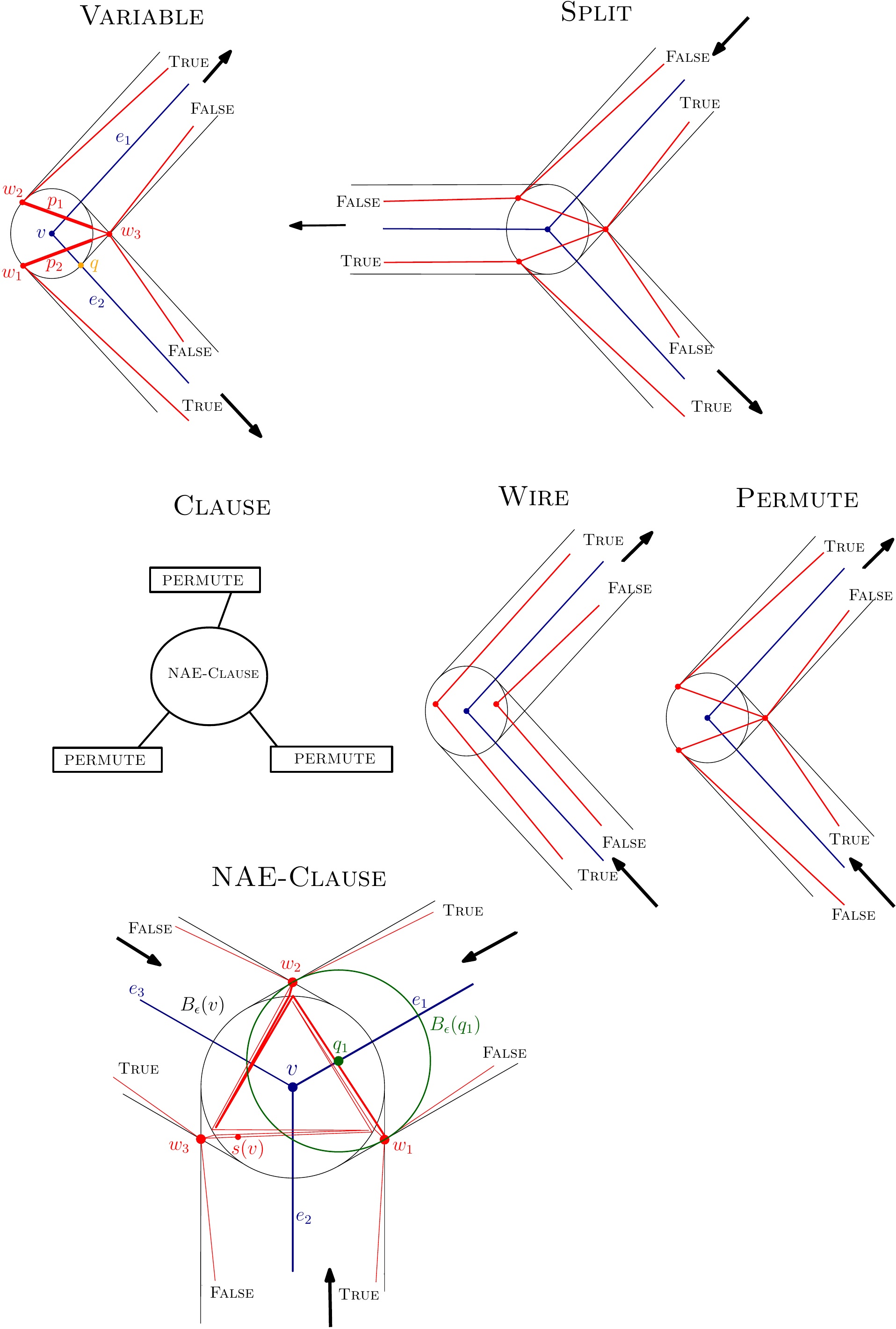}
\caption{Building blocks to build graph-similarity instance given a \textsc{Monotone-Planar-3-Sat} instance.}
\label{fig:gadgets}
\end{figure}

In the following, we give a detailed description of the construction
and the properties of the gadgets.  The \textsc{Variable} gadget: For
each variable of $A$, we add a vertex $v$ and two edges, $e_1$ and
$e_2$, of $G_1$ incident to a vertex $v$, where the angle between
$e_1$ and $e_2$ is between $90^{\circ}$ and $120^{\circ}$. We add a
vertex $w_1$ ($w_2$) of $G_2$ on the intersection of the outer
boundaries of $T_{\eps}(e_2)$ ($T_{\eps}(e_1)$) and a line through
$e_1$ ($e_2$). Furthermore, we add a vertex $w_3$ of $G_2$ at the
intersection of the boundaries of $T_{\eps}(e_1)$ and
$T_{\eps}(e_2)$. See the upper left sketch in Figure \ref{fig:gadgets}
for an illustration. We connect $w_1$ and $w_2$ with $w_3$ and draw an
edge from $w_1$ and $w_2$ inside the $\eps$-tubes around $e_1$ and
$e_2$, with labels \textsc{True}. Analogously, we embed two edges from
$w_3$ with the label \textsc{False}. For the \textsc{Variable} gadget
a \textsc{True}-\textsc{True} combination is not possible: There are
two placements $p_1$ and $p_2$ of the vertex $v$. Assume we choose
$p_1$. Note that one can map $e_1$ to a path containing the edge of
$G_2$ with the \textsc{True} labeling inside $T_\eps(e_1)$. Now, we
want to map $e_2$ to a path $P$ starting at some point of $p_1$, where
$P$ contains the edge of $G_2$ with the \textsc{True} labeling inside
$T_\eps(e_2)$. In this case, one has to walk along $e_2$ up to $q$
(the point on $e_2$ with distance $\eps$ to $w_3$) while walking
simultaneously to $w_3$ on $P$. But then, when walking along $P$ up to
$w_1$, one must walk back along $e_2$ up to $v$ as any point along the
interior of $e_1$ has distance greater than $\eps$ to $w_1$. Thus,
$\delta_F(e_2, P) > \eps$. It is easy to see that for any other
combination of labels there is a placement $p$ of $v$, such that $e_1$
and $e_2$ can be mapped to a path $P_1$ ($P_2$) starting at a point
$a$ of $p$ with $\delta_F(e_2, P_1) \leq \eps$ ($\delta_F(e_2, P_2)
\leq \eps$).

A \textsc{Permute} gadget is the same as the \textsc{Variable} gadget, but with a different labeling, see Figure \ref{fig:gadgets}.
The construction of the \textsc{Split} gadget is similar to the \textsc{Variable} gadget. Additionally, we add a third edge $e_3$ of $G_1$ and edges of $G_2$ from $w_2$ and $w_3$ inside the $\eps$-tube around $e_3$. For the labeling, see Figure~\ref{fig:gadgets}.

The same argument as for the \textsc{Variable} gadget is used to proof
the following statements:
\begin{itemize}
\item A \textsc{False} signal can never be converted to a \textsc{True} signal in the \textsc{Split} gadget.
\item A \textsc{False} signal can never be converted to a \textsc{True} signal in the \textsc{Permute} gadget.
\end{itemize}

Furthermore, a \textsc{True} signal can, but need not to be converted
to a \textsc{False} signal in the \textsc{Permute} gadget.

We construct the \textsc{Wire} gadget used to connect all the other
gadgets by drawing two edges $e_1$ and $e_2$ of $G_1$ incident to a
vertex $v$ (with arbitrary angle) and two vertices $w_1$ and $w_2$ of
$G_2$ inside $B_{\eps}(v)$ with non intersecting incident edges inside
$T_{\eps}(e_1)$ and $T_{\eps}(e_2)$. Obviously, it is not possible to
convert a \textsc{True} signal to a \textsc{False} signal or vice
versa, here.

For the \textsc{Clause} gadget, we first introduce a
\textsc{NAE-Clause} gadget. Here it is required that the three values
in each clause are not all equal to each other. We start the
construction of the \textsc{NAE-Clause} gadget by drawing three edges
$e_1$, $e_2$ and $e_3$ incident to a vertex $v$ with a pairwise
$120^{\circ}$ angle. We draw three vertices $w_1$, $w_2$ and $w_3$ on
the intersections of $T_{\eps}(e_1)$, $T_{\eps}(e_2)$ and
$T_{\eps}(e_3)$ with a maximum distance to $v$. Furthermore, we draw
two edges of $G_2$ inside the tubes for each vertex and label them as
shown in Figure~\ref{fig:gadgets}. Let $q_1$ be the point on $e_1$
with distance $\eps$ to $w_1$ and $w_2$. Here, it is not enough to
simply connect $w_1$ and $w_2$ with one edge as shown in the bottom
left sketch of Figure~\ref{fig:gadgets}. To force backtracking along
$e_1$ for a combination of labels which we want to exclude, we have to
ensure that a path from $w_1$ to $w_2$ leaves $B_\eps(q_1)$ but stays,
once entered, inside $B_\eps(v)$. For the other pairs, $(w_1, w_3)$
and $(w_2, w_3)$ we do the same. A possible drawing of these paths
maintaining the planarity of $G_2$ is shown in \figref{fig:gadgets}.
The three placements of $v$ are connected by the vertices $w_1$, $w_2$
and $w_3$ and it holds that there is no placement of $v$ such that an
all-\textsc{False} or an all-\textsc{True} labeling can be realized:
Suppose we map $v$ to $s(v)$ as shown in the
\figref{fig:gadgets}. Then, edges $e_2$ and $e_3$ can be mapped to
paths through edges labeled \textsc{True}. But we cannot map $e_1$ to
such a path $P$: When $P$ reaches vertex $w_1$, any corresponding
reparameterization of $e_1$ realizing $\delta_F(e_1,P)\leq\eps$ must
have reached $q_1$ as $q_1$ is the only point with distance at most
$\eps$ to $w_1$ on $e_2$. As $P$ leaves $B_\eps(q_1)$ between $w_1$
and $w_2$ and any point on $e_1$ with distance at most $\eps$ to the
part of $P$ outside $B_\eps(q_1)$ lies between $v$ and $q_1$ it
follows that $\delta_F(e_1,P)>\eps$. For symmetric reasons it follows
that any other all-equal labeling cannot be realized. However, there
is a placement of $v$, such that all three edges $e_1$, $e_2$ and
$e_3$ can be mapped to a path in $G_2$ with \fr\ at most $\eps$, for
each configuration where not all three signals have the same value.

\textsc{Monotone-Planar-NAE-3-Sat} is in $P$, but we can use the
\textsc{NAE-Clause} gadget as core of our \textsc{Clause} gadget
referring to the NP-complete version \textsc{Monotone-Planar-3-Sat}:
We obtain the \textsc{Clause} gadget by connecting each
\textsc{NAE-Clause} gadget with three \textsc{Permute} gadgets, as
shown in Figure~\ref{fig:gadgets}.

\figref{fig:graph_construction} partially shows the constructed graphs for a given \textsc{Monotone-Planar-3-Sat} $A$ consisting of the subgraphs (gadgets) described above.

\begin{figure}
\centering
\includegraphics[scale=0.745]{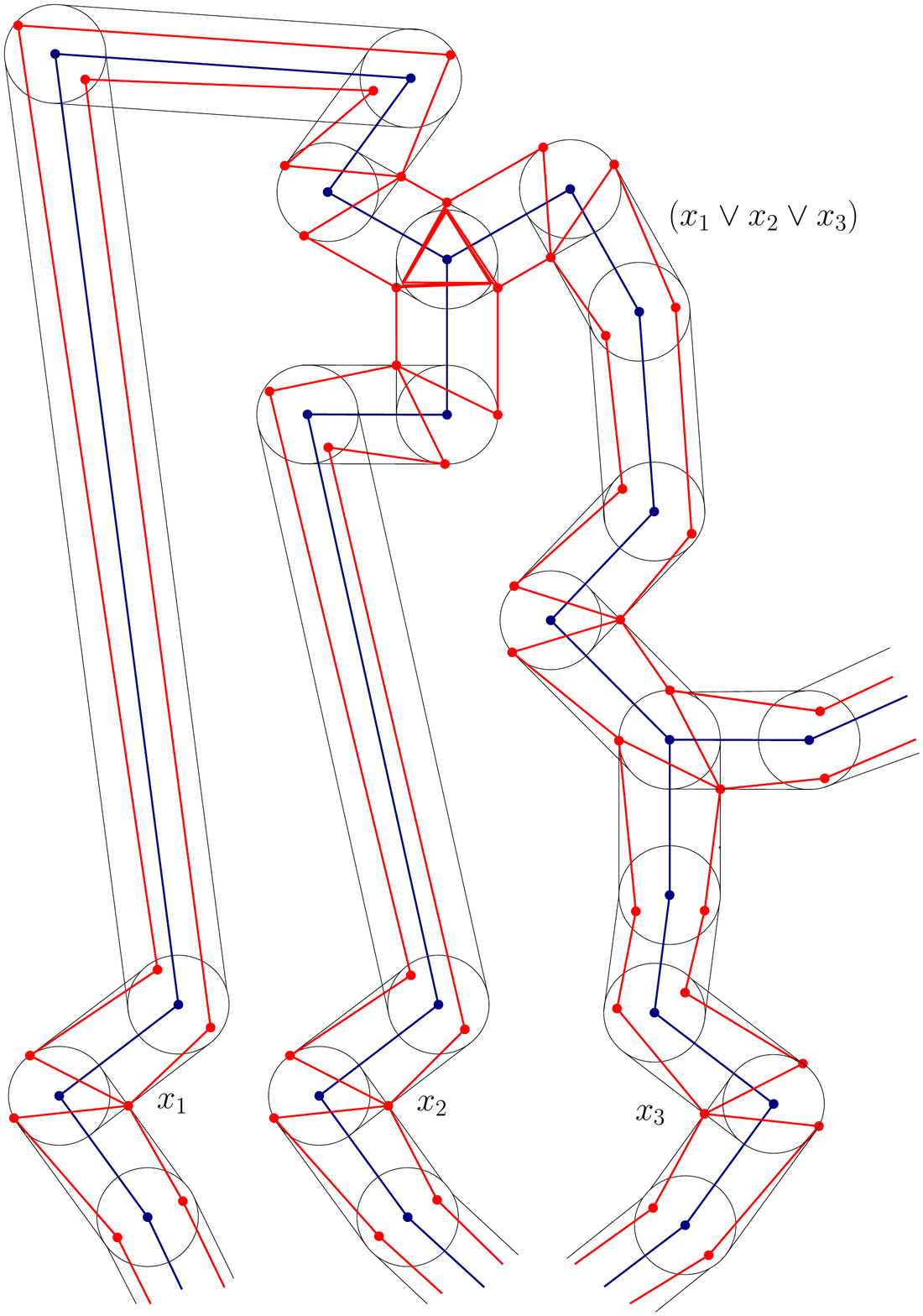}
\caption{For the \textsc{Monotone-Planar-3-Sat} instance $A$ with
  variables $V=\{x_1,x_2,\dots,x_5\}$ and clauses $C=\{(x_1 \vee x_2
  \vee x_3), (x_3 \vee x_4 \vee x_5), (\bar{x_1} \vee \bar{x_3} \vee
  \bar{x_5})\}$ the Figure shows the construction of the clause $(x_1
  \vee x_2 \vee x_3)$.}
\label{fig:graph_construction}
\end{figure}

Now, given a \textsc{Monotone-Planar-3-SAT} instance $A$, one can
construct the graphs $G_1$ and $G_2$ with the gadgets described
above. Note that all gadgets are plane subgraphs. By placing them next
to each other with no overlap, we can ensure that $G_1$ and $G_2$ are
plane graphs.

A valid placement of the whole graph $G_1$ induces a solution of $A$:
In the corresponding gadget for each positive NAE-clauses, at least
one of the outgoing edges of $G_1$ must be mapped to a path through an
edge labeled \textsc{True}. By construction, this label cannot be
converted to \textsc{False} in any of the gadgets and therefore the
corresponding variable $v$ gets the value \textsc{True}. In this case,
$v$ cannot set any of the negative clauses \textsc{True} because the
other outgoing edge must be mapped to a path through the edge of $G_2$
labeled \textsc{False} and this signal can never be switched to
\textsc{True}. The same holds for the case of negative NAE-clauses.

Conversely, given a solution $S$ of the \textsc{Monotone-Planar-3-SAT}
instance $A$, it is easy to construct a placement of $G_1$. In the
variable gadget of variable $x_1$, we chose placement $p_1$ and map
$e_1$ to a path through the edge of $G_2$, labeled \textsc{True} and
map $e_2$ to a path through the edge labeled \textsc{False} if $x_1$
is positive in $S$. If $x_1$ is negative, we chose placement $p_2$ and
map $e_1$ and $e_2$ accordingly.  All edges of the other gadget now
can be mapped to $G_2$ in a signal preserving manner (\textsc{True}
stays \textsc{True}, \textsc{False} stays \textsc{False}). If there
exists a clause $C$ in $A$, such that all three variables of $C$ are
positive (negative) in $S$, we change one signal in the
\textsc{Permute} gadget from \textsc{True} to \textsc{False}. Thus, we
have found a placement for the whole graph $G_1$.
\end{proof}

The following stronger result follows from the observation that the
characteristics of the subgraphs we constructed in the proof of
\thmref{theorem:NP_plane} still hold for a slightly larger $\eps$
value.
\begin{theorem}
It is NP-hard to approximate $\distDir(G_1, G_2)$ within a $1.10566$ factor.
\label{thm:c_hard}
\end{theorem}
\begin{proof}
\begin{figure}
\centering
\includegraphics[scale=0.65]{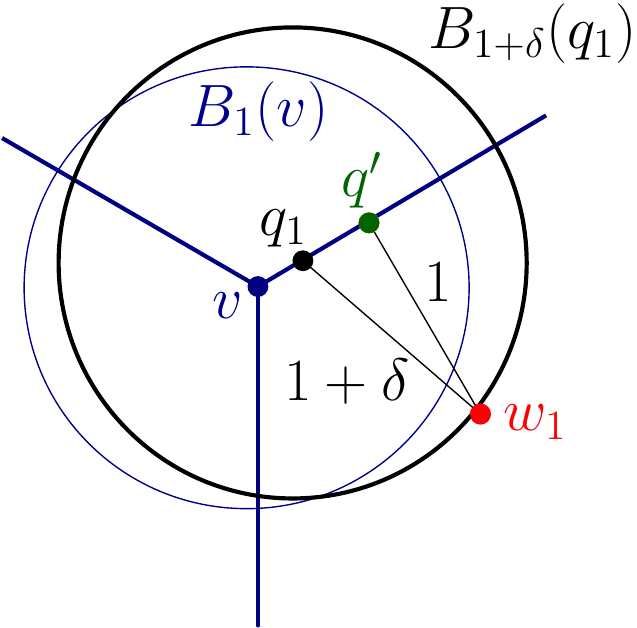}
\caption{Illustration of the proof of Theorem~\ref{thm:c_hard}.}
\label{fig:approximation}
\end{figure}
We give a detailed proof for the \textsc{NAE-Clause} gadget and note that a
similar argument holds for the other gadgets. See
Figure~\ref{fig:approximation} for an illustration of the arguments
and the calculations below.

Let us fix $\eps = 1$. As described in the proof of
Theorem~\ref{theorem:NP_plane}, we connect the vertices $w_1$ and
$w_2$ by a path which leaves $B_{1}(q_1)$, but stays inside
$B_{1}(v_1)$. (Introducing a spike which leaves $B_{1}(q_1)$ and
returns to $B_{1}(q_1)$, see Figure~\ref{fig:gadgets}). We draw the
spike such that its peak is arbitrarily close to the intersection of a
straight line through the edge $e_1$ and the $1$-circle around
$v$. When enlarging $\eps$, the point $q_1$ moves along $e_1$ toward
$v$. We need to compute the smallest value $\delta_{min}$, such that
$B_{1}(v)$ is completely contained in $B_{1+\delta_{min}}(q_1)$. For
any value $\delta < \delta_{min}$, there exists a drawing of the
spikes, such that the characteristics of the \textsc{NAE-Clause}
gadget still hold, e.g., there is no placement of $v$ allowing an
all-equal-labeling.

Note that $\delta_{min}$ equals the distance from $q_1$ to $v$, when
$q_1$ is at distance $1+\delta_{min}$ to $w_1$. Let $q'$ be the
position of $q_1$ for $\delta=0$ and let $d$ be the distance between
$q'$ and $q_1$. Then we have
$
\tan(30^{\circ}) = \frac{\delta_{min} + d}{1} =\delta_{min} + d
$. 
Furthermore, we have
$
d = \sqrt{(1+\delta_{min})^2 -1}
$ 
and therefore
$
\delta_{min} = \tan(30^{\circ}) - \sqrt{(1+\delta_{min})^2 -1}
$, 
which solves to
$\delta_{min} = \frac{1}{4} - \frac{1}{4\sqrt{3}} \approx 0.10566 $.
The factor by which $\eps$ can be multiplied is greater than
$1+\delta_{min}$ for all other gadgets. Thus, $\delta_{min}$ is the
critical value for the whole construction and the theorem follows.
\end{proof}

\subsection{Deciding the Strong Graph Distance in Exponential Time}
\label{sec:exponential_time_algo}
A brute-force method to decide the directed strong graph distance is
to iterate over all possible combinations of valid vertex
placements, which takes $O(m_1 \cdot m_2^{n_1})$. time.
Another approach is to decompose $G_1$ into faces and merge the 
substructures bottom-up.

First, we remove all tree-like substructures of $G_1$ and map these as
described in the proof of Lemma~\ref{lem:gd-map-tree}. Next, we
decompose the remainder of $G_1$ into chordless cycles, where a chord
is a maximal path in $G_1$ incident to two faces, see
Figure~\ref{fig:gd-plane-decomposition}.
\begin{figure}[b]
\centering
\includegraphics[page=1, width=0.7\textwidth]{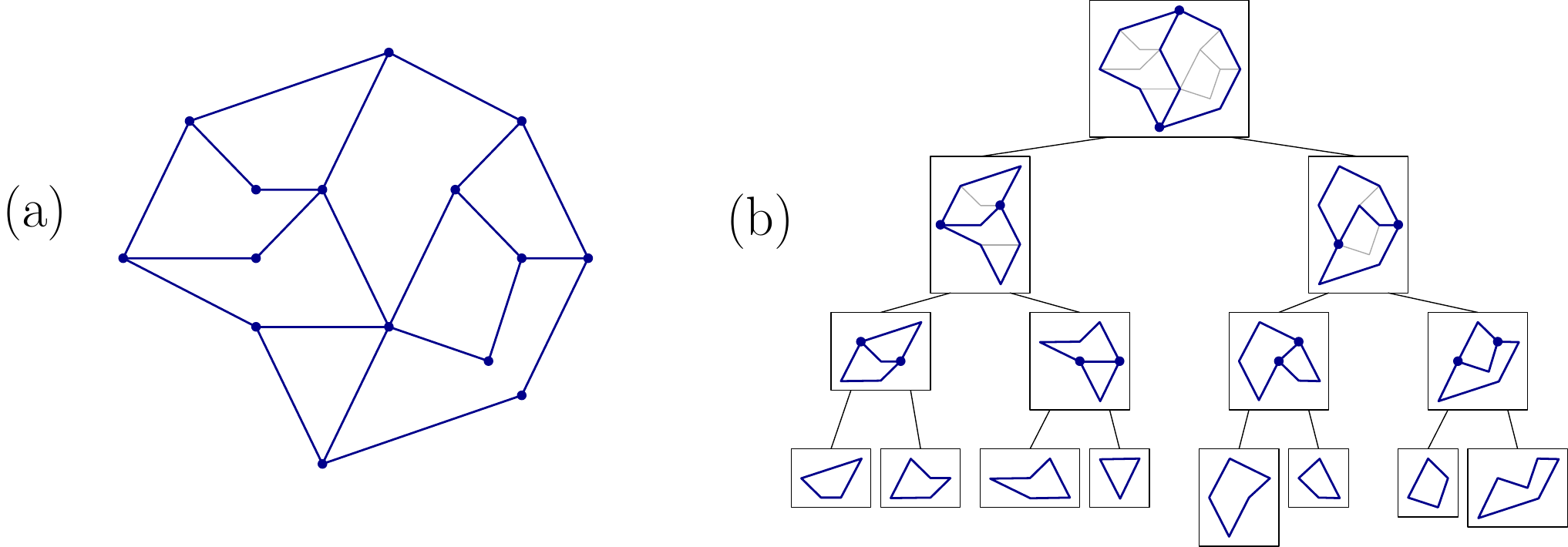}
\caption{A plane graph (a) is recursively decomposed into chordless cycles by splitting each cycle with a chord (b).}
\label{fig:gd-plane-decomposition}
\end{figure}
We map the parts of $G_1$
from bottom up, deciding in each step if we can map two adjacent
cycles and all the nested substructures of the cycles
simultaneously. To do so, we start with storing all combinations of
placements of endpoints of a chord -which separates two faces
(chordless cycles)- allowing us to map the two faces
simultaneously. We prune all placements which are not part of any
valid combination. In the following steps, for each placement $C_u$ of
an endpoint $u$ of a chord and each valid combination $c$ of nested
chords computed in the previous step, we run one graph
exploration. For each placement $C_v$ of the other endpoint $v$ of the
chord, which allow to map both cycles simultaneously, we store a new
combination consisting of $C_u$, $c$ and $C_v$. We prune all
placements of $u$ where we cannot reach a valid placement of $v$ by
using any of the previous computed combinations. Furthermore, we prune
those placements of $v$ which are never reached by any graph
exploration. If the list of placements gets empty for one vertex, we
can conclude, that the graph distance is greater than
$\eps$. Conversely, if we find a valid combination of placements of
the endpoints of the chord in the last step, we can conclude that we
can map the whole graph $G_1$ as we guarantee in each step that all
substructures can be mapped, too.

\begin{theorem}
\label{thm:exponential_time}
For plane graphs, the strong graph distance can be decided in $O(F
m_2^{2F-1})$ time and $O(m_2^{2F-1})$ space, where $F$ is the number
of faces of $G_1$.
\end{theorem}
\begin{proof}
Each graph exploration takes $O(m_2)$ time and in each node we have to
run $O(m_2k)$ explorations, where $k$ is the number of valid
combinations of endpoint placements from previously investigated
chords. As the tree has a depth of $\log(F)$, we have
\begin{ceqn}
\begin{align*}
 \frac{F}{2} m_2^2 + \frac{F}{4} m_2^2m_2^4 + \dots + m_2^2\left(m_2^{2^{log(F)-2}}\right)^2
 = F \sum_{i=1}^{\log(F)} \frac{1}{2^i} m_2^{2^{i+1}-2}
\end{align*}
\end{ceqn}
as the total running time of the graph explorations. We can upper bound this term as follows:
\begin{ceqn}
\begin{align*}
&F \sum_{i=1}^{\log(F)} \frac{1}{2^i} m_2^{2^{i+1}-2}
\leq F \sum_{i=1}^{\log(F)} m_2^{2^{i+1}-2}
\leq F \sum_{i=1}^{2^{\log(F)+1}-2} m_2^{i}\\
&= F \sum_{i=1}^{2F-2} m_2^{i}
= F \frac{m_2^{2F-1}-m_2}{m_2-1} \in O\left(F m_2^{2F-1}\right),
\end{align*}
\end{ceqn}
where the second equality uses
\begin{ceqn}
\begin{align*}
\sum_{i=1}^n a^i = \frac{a(a^n - 1)}{a-1},
\end{align*}
\end{ceqn}

for $a \in \mathbb{R}$.  In the first step, we have to store
$O(m_2^2)$ combinations for each two faces we want to map
simultaneously.  Let $k$ be the number of combinations in the previous
step. Then we have to store up to $k^2m_2^2$ combinations in the next
step. This results in storing $O(m_2^{2F-1})$ combinations in the root
node of the decomposition.
\end{proof}

Thus, this method 
is superior to the brute-force method if $2F-1 \leq n_1$.

\subsection{Approximation for Plane Graphs}
For plane graphs, \algoref{alg:general} yields an approximation depending on the angle between the edges
for deciding the strong graph distance. The decision is based on
the existence of valid placements. Therefore, the runtime is the same
as stated in Theorem \ref{thm:gd-strong-dist}.

\begin{theorem}
Let $G_1:=(V_1, E_1)$ and $G_2:=(V_2, E_2)$ be plane graphs.  Assume
that for all adjacent vertices $v_1$, $v_2 \in V_1$, $B_\eps(v_1)$ and
$B_\eps(v_2)$ are disjoint. Let $\alpha_v$ be the smallest angle
between two edges of $G_1$ incident to vertex $v$ with $deg(v) \geq
3$, and let $\alpha := \frac{1}{2} \min_{v\in V_1}(\alpha_v)$.  If
there exists at least one valid $\eps$-placement for each vertex of
$G_1$, then $\distDir(G_1, G_2) \leq \frac{1}{\sin(\alpha)} \eps$.
\label{thm:alpha-approximation}
\end{theorem}
\begin{proof}
\begin{figure}[t]
\centering
\includegraphics[scale=0.6]{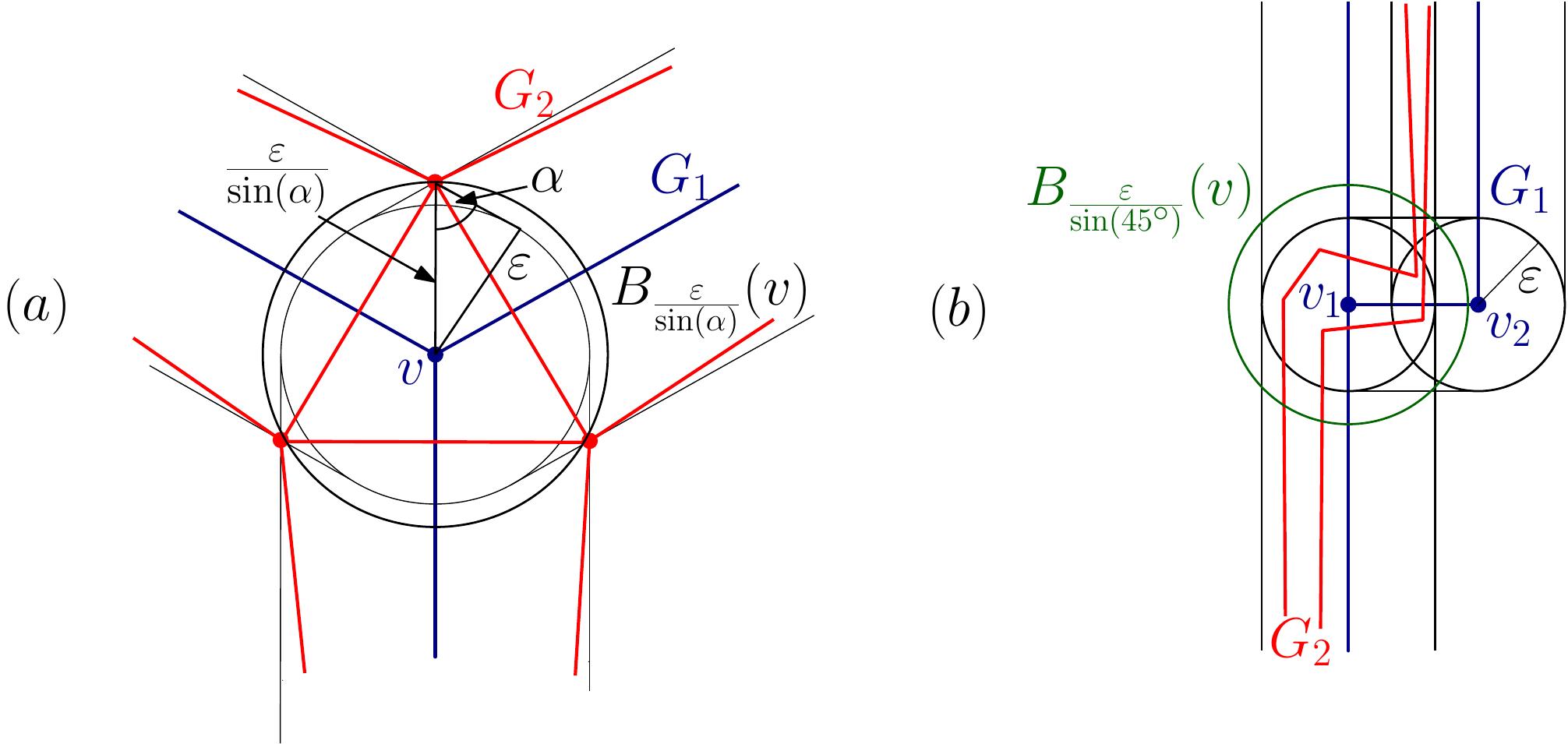}
\caption{Illustration of the proof of
  Theorem~\ref{thm:alpha-approximation}. (a) Three valid
  $\eps$-placements form one connected component inside a slightly
  larger ball around $v$. (b) In case of short edges, the placements
  remain unconnected also inside a larger surrounding of $v$. }
\label{fig:alpha-approximation}
\end{figure}
Let $\alpha$ be the smallest angle between two edges incident to a
vertex $v$ with degree at least three and let $C_1,C_2,\dots, C_j$ be
the valid placements of $v$ for a given distance value
$\eps$. Furthermore, let $V_{C_i}$ be the set of vertices of $C_i$. It
can be easily shown that for a larger distance value of $\hat{\eps}
\geq \frac{1}{\sin(\alpha)}\eps$ there exists vertices $v_1,v_2,\dots,
v_k$, embedded inside $B_{\hat{\eps} }$, such that the subgraph $C =
(V',E')$, where $V' = \bigcup\limits_{i=1}^{j} V_{C_i} \cup \{v_1\}
\cup \{v_2\} \cup \dots \cup \{v_k\}$ and $E'=\{\{uw\} \in E_2| u \in
V', w \in V'\}$ is connected, see
Figure~\ref{fig:alpha-approximation}~(a). Note that this property is
not true if $B_\eps(v_1) \cap B_\eps(v_2) \neq \emptyset$ for two
adjacent vertices $v_1$, $v_2 \in V_1$:
Figure~\ref{fig:alpha-approximation}~(b) shows an example with
$\alpha=45^\circ$ where the two placements do not merge inside
$B_{\frac{\varepsilon}{\sin(45^\circ)}}(v)$. However, with the
condition $B_\eps(v_1) \cap B_\eps(v_2) = \emptyset$, there is only
one valid $\frac{1}{\sin(\alpha)} \eps$-placement $C$ for each vertex
with degree at least three. Furthermore, every valid $\eps$ placement
is a valid $\frac{1}{\sin(\alpha)} \eps$-placement. Now, a path $P$ of
$G_1$ starting at a vertex $v$ with $deg(v)\geq 3$ and ending at a
vertex $w$ with $deg(w)\neq 2$, with vertices of degree two in the
interior of $P$, can be mapped as described in the proof of
Lemma~\ref{lem:gd-map-tree} . For two paths which start and/or end at
a common vertex $v$, $v$ is mapped to the same placement as there is
only one valid $\frac{1}{\sin(\alpha)} \eps$-placement of $v$. This
ensures that each edge of $G_1$ is mapped correctly.
\end{proof}

\section{Experiments on Road Networks}
\label{sec:experiments}
In the last decade several algorithms have been developed for
reconstructing maps from the trajectories of entities moving on the
network~\cite{akpw-cemca-15,akpw-mca-15}. This naturally asks to
assess the quality of such reconstruction algorithms.  Recently, Duran
et al~\cite{duran} compared several of these algorithms on hiking
data, and found that inconsistencies often arise due to noise and low
sampling of the input data, for example unmerged parallel roads or the
addition of short off-roads.

When assessing the quality of a network reconstruction from trajectory
data, several aspects have to be taken into account.  Two important
aspects are the \emph{geometric} and \emph{topological error} of the
reconstruction. Another important aspect is the \emph{coverage}, i.e.,
how much of the network is reconstructed from the data. We believe our
measures to be well suited for assessing the geometric error while
still maintaining connectivity information.

We have used the weak graph distance for measuring the distance
between different reconstructions and a ground truth, as well as a
simplification of part of the road network of Chicago.
\figref{fig:chicago_map_reconstruction}~(a) shows two reconstructed
road map graphs $R$ (red) and $B$ (blue), overlayed on the underlying
ground truth road map $G$ from OpenStreetMap.  The reconstruction $R$
in red resulted from Ahmed et al.'s algorithm \cite{Ahmed:2012},
whereas the reconstruction $B$ in blue from Davies et al.'s
\cite{Davies} algorithm. 
%
Our directed graph distance from $B$ to $G$ is 25 meters, and from $R$
to $G$ it is 90 meters.  This reflects the local geometric error of
the reconstructions (note that it does not evaluate the difference in
coverage). \figref{fig:chicago_map_reconstruction}~(b) shows an
example where the topology of $R$ and $G$ differs (blue circle),
affecting for instance navigation significantly. Our measure captures
this difference. Although the reconstruction approximates the geometry
well, our measure computes a directed distance of 200 m from $G$
(restricted to the part covered by $R$) to $R$.

\begin{figure}
\centering
\subfloat[Two partial map reconstructions of Chicago.]{
  \includegraphics[width=0.67\textwidth]{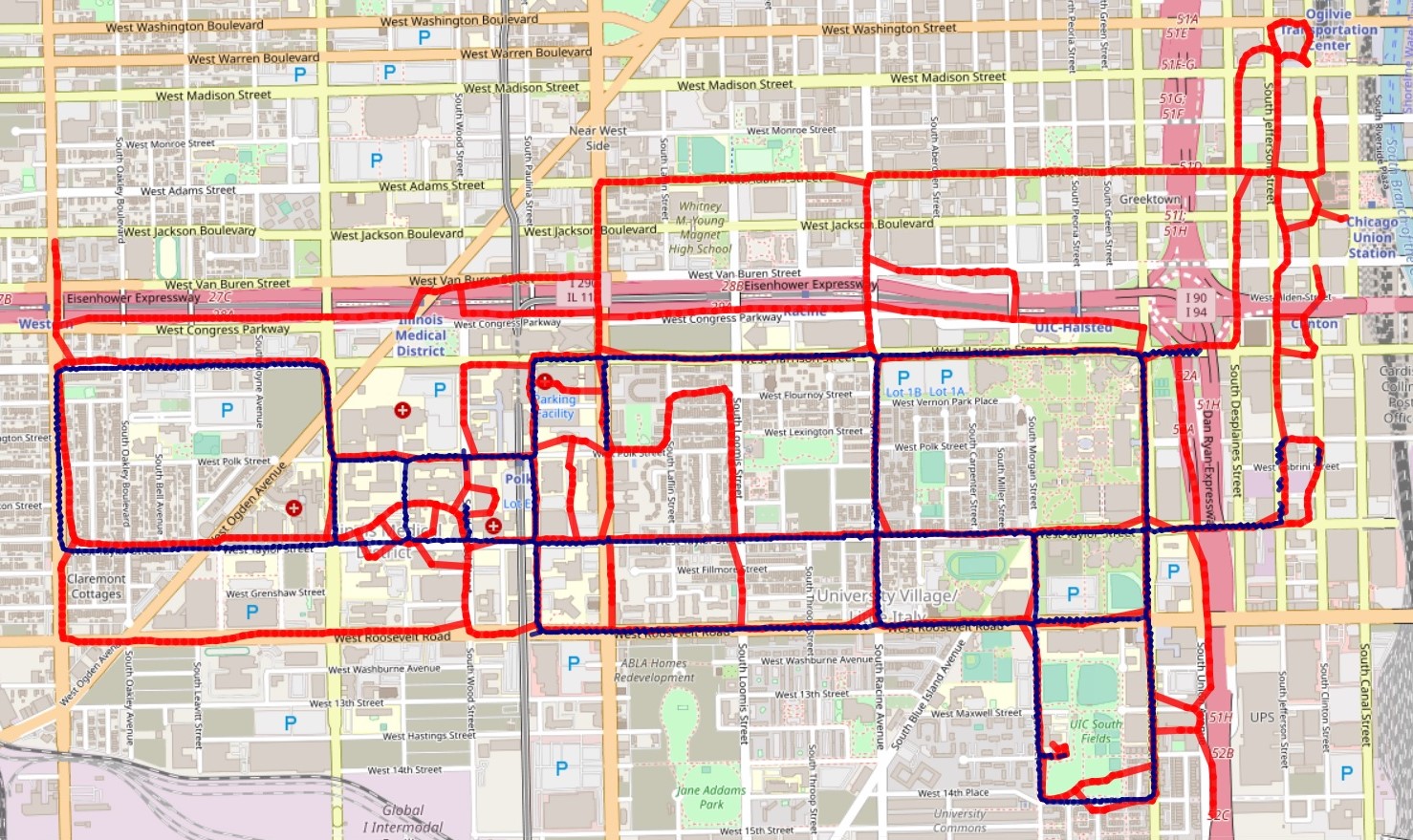}}
  \label{subfig:chicago_map_reconstruction}
  \hfill
  \subfloat[Different topology.]{
  \includegraphics[width=0.3\textwidth]{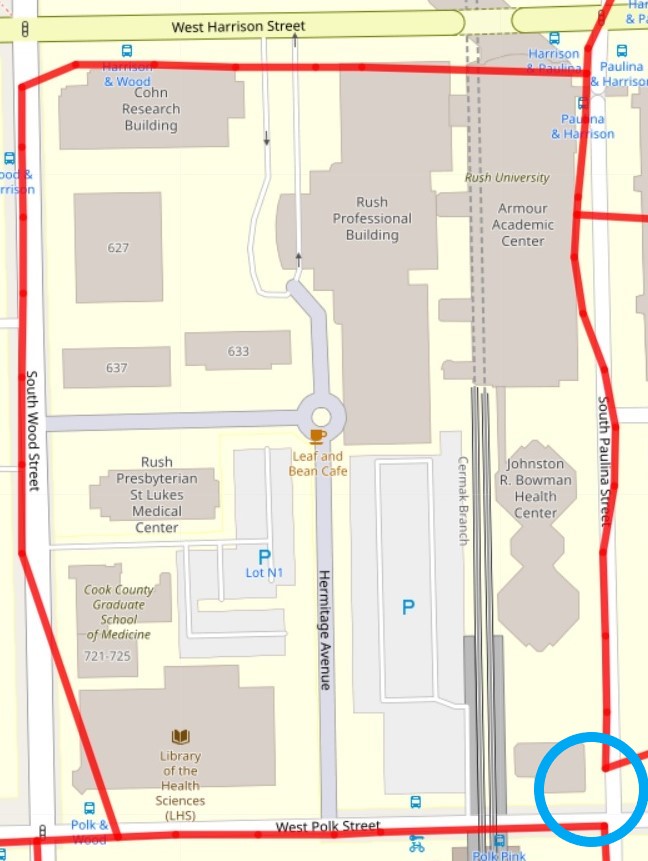}}
  \label{subfig:chicago_map_reconstruction_part}
  \caption{Two reconstructed road map graphs $R$ (red) and $B$ (blue), overlayed on the underlying ground truth road map $G$ from OpenStreetMap.}
\label{fig:chicago_map_reconstruction}
\end{figure}

Figure~\ref{fig:chicago_simplification} shows a partial road network
of Chicago at different resolutions. Both maps show vertices as blue
dots, connected by straight-line edges. Our approach yields a distance
of 22 meters between the graphs, which corresponds to the geometric
error of the lower resolution map in comparison to the higher
resolution map. The data is extracted from OpenSteetMap using the
OSMnx Python library.

\begin{figure}
\subfloat[Partial road network of Chicago.]{
  \includegraphics[width=0.54\textwidth]{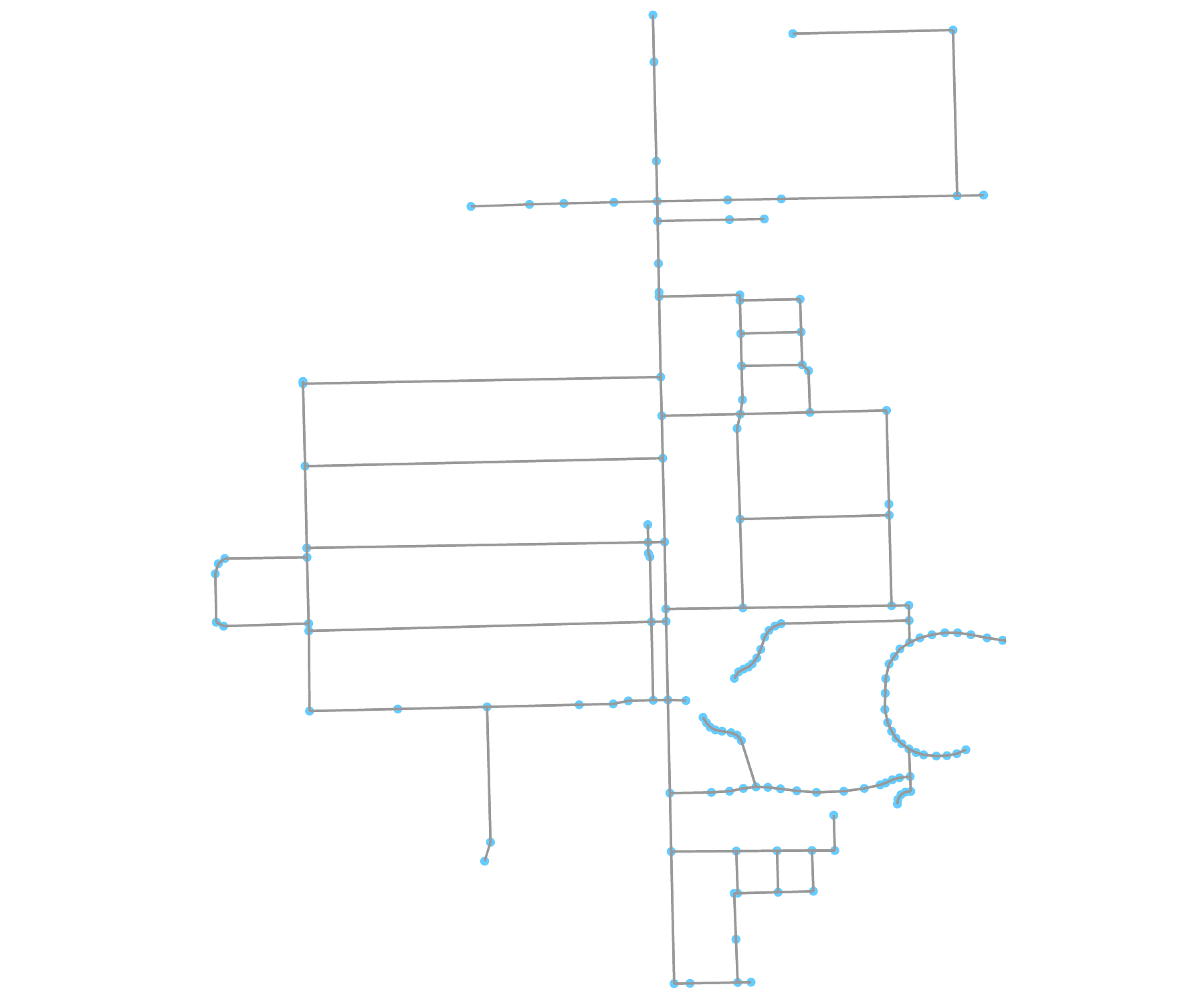}}
  \label{subfig:map_chicago}
  \hfill
  \subfloat[Simplified partial road network of Chicago.]{
  \includegraphics[width=0.44\textwidth]{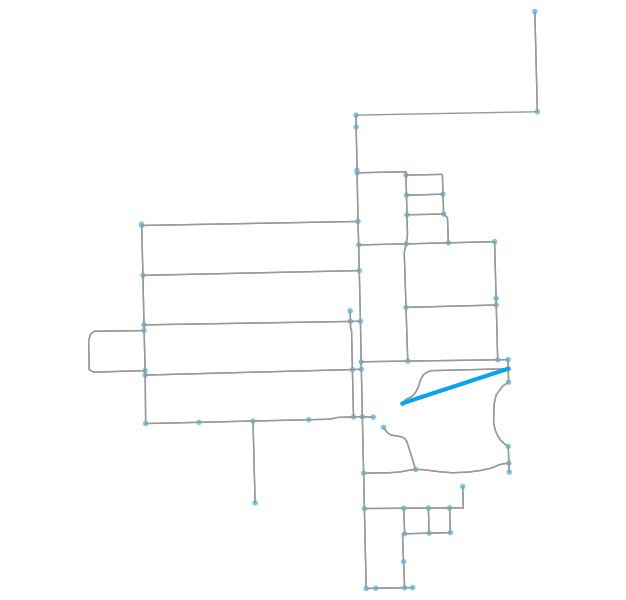}}
  \label{subfig:map_chicago_simple}
	\caption{The graph $G$ in (a) is a higher resolution map, while the graph $H$ in (b) is a lower resolution map that represents the road segment geometries with fewer vertices. Note that edges are all embedded as straight line segments.}
\label{fig:chicago_simplification}
\end{figure}

\clearpage
\section{Conclusion}
We developed new distances for comparing straight-line embedded graphs
and presented efficient algorithms for computing these distances for
several variants of the problem, as well as proving NP-hardness for
other variants.  Our distance measures are natural generalizations of
the \fr and the \wfr to graphs, without requiring the graphs to be
homeomorphic.  Although graphs are more complicated objects than
curves, the runtimes of our algorithms are comparable to those for
computing the \fr\ between polygonal curves.  A large-scale comparison
of our approach with existing graph similarity measures is left for
future work.



\bibliographystyle{plainurl}
\bibliography{graph-similarity}





%
%
\end{document}